\documentclass[11pt]{article}

\usepackage[margin=1in]{geometry}
\usepackage[T1]{fontenc}
\usepackage[utf8]{inputenc}
\usepackage{lmodern}
\usepackage{authblk}
\usepackage[numbers,sort&compress]{natbib}
\usepackage{amsmath,amssymb,amsthm,amsfonts}

\usepackage{graphicx,epstopdf}
\usepackage{booktabs}
\usepackage{multirow}
\usepackage{array}
\usepackage{float}
\usepackage[colorlinks=true,linkcolor=blue,citecolor=blue,urlcolor=blue]{hyperref}
\usepackage[numbers]{natbib}

\usepackage{caption}
\usepackage{subcaption}
\usepackage{enumitem}

\newtheorem{lemma}{Lemma}

\theoremstyle{definition}

\theoremstyle{remark}

\title{Resilience Beyond Pairwise Networks}

\author[1]{Amitosh Tiwari}
\author[3,4]{Chittaranjan Hens}
\author[1,2]{Prosenjit Kundu}

\affil[1]{CSys Lab, Complex System Group, Dhirubhai Ambani University, Gandhinagar, Gujarat 382007, India}
\affil[2]{Smart Energy Learning Centre (SELC), Dhirubhai Ambani University, Gandhinagar, Gujarat 382007, India}
\affil[3]{Center for Computational Natural Sciences and Bioinformatics, International Institute of Information Technology Hyderabad, Gachibowli, Hyderabad 500032, Telangana, India}
\affil[4]{Biomedical Research Center, International Institute of Information Technology Hyderabad, Gachibowli, Hyderabad 500032, Telangana, India}

\begin{document}
\maketitle
\begin{abstract}
%Low-dimensional descriptions of network dynamics are constructed for a single interaction order either pairwise or purely higher order interaction. This leaves an important gap for systems in which pairwise links and group interactions act together. 
We derive a one-dimensional reduction for nonlinear dynamics on simplicial complexes containing both pairwise and triangular (higher-order) interactions. The effective state is defined using a mixed weight determined by the pairwise and triangular degrees of each node. The resulting reduced equation retains two structural coefficients, associated separately with the pairwise and higher-order coupling channels. A fluctuation expansion identifies the closure assumptions underlying the reduction and shows how deviations of individual node states from the effective state contribute to the approximation error. We numerically validate the proposed framework on  Gene-regulatory dynamics, the double-well system, and SIS spreading. The states of the reduced model are compared with full-network simulations through coupling-parameter sweeps, steady-state branch calculations, and progressive node-removal experiments on synthetic and real-world networks. The reduced model successfully reproduces the principal transitions and steady-state branches in all three dynamical systems considered. Agreement is strongest for relatively homogeneous networks and deteriorates when structural heterogeneity produces a broader distribution of node states. The closure diagnostics account for this loss of accuracy and indicate when a single effective state is no longer sufficient. The reduction therefore provides a tractable description of resilience in systems with coexisting pairwise and higher-order interactions.
\end{abstract}
\textbf{keywords :}
Higher-order interactions,
Simplicial complexes,
Network resilience,
Effective-state reduction, 
Nonlinear dynamics, 
Critical transitions.
% --- Main Text ---
\section{Introduction}
\label{sec:introduction}

Resilience is the ability of a system to retain its functioning under disturbances, failures, or structural changes
\cite{holling1973resilience,walker2004resilience,scheffer2009critical,
krakovska2024resilience,schoenmakers2021resilience}.
In networked systems, resilience is closely related to the stability of the dynamical process evolving on the underlying interaction structure
\cite{gao2016universal,artime2024robustness,liu2022network}.
Its loss may lead to a sudden transition between qualitatively different states, with consequences for ecological, technological, and socio-economic systems
\cite{holling1973resilience,gao2016universal}.
A key question is therefore how the resilience of a large interacting system arises from the interplay between local nonlinear dynamics and network organisation.
This question arises naturally in network science. Dynamical systems on complex networks have been used to study synchronization
\cite{pikovsky2001synchronization,arenas2008physrep,ji2013prl,
rodrigues2016kuramoto,kundu2017pre,kundu2018epl,kundu2019chaos,
khanra2018pre,khanra2021csf,dutta2025hypergraph,dutta2023perfect,
das2025phaselag,dutta2023phase,dutta2024adaptive,dutta2025double,
ghosh2025universal},
epidemic spreading and diffusion
\cite{pastor2001prl,granell2013prl,pastor2015rmp,wang2017rpp,
mei2017annualreview,colizza2007bmc,higham2021epidemics,
yuan2026noise,luo2026temporal},
and many other collective phenomena
\cite{albert2002rmp,barrat2008dynamical,newman2010networks,
boccaletti2006physrep,dorogovtsev2008rmp}.
Together, these studies show that the organisation of interactions can alter the stability, controllability, and response of a system to perturbations
\cite{dawn2026instability,Meena2023stability,Allesina2012}.

\par
A persistent challenge in the analysis of networked dynamical systems is the dimensionality of the governing equations. A network of $N$ interacting units generally gives rise to $N$ coupled nonlinear equations, making the direct computation of equilibria, stability boundaries, and bifurcation analysis increasingly difficult as the system size grows. Dimension-reduction methods address this difficulty by replacing the original dynamics with one or a few effective variables
\cite{gao2016universal,kundu2022accuracy,jiang2018predicting,
laurence2019prx,Tu2024Separable,Wu2024Collapse}.
By reducing the number of governing variables, these approaches make the analysis of steady states, bifurcation points, and tipping behaviour more tractable
\cite{MaclarenJROS2023,Duan2025Multilayer}.
A prominent example is the effective-state formulation of Gao, Barzel, and Barab\'asi, which showed that the collective behaviour of a broad class of networked systems can be approximated by a one-dimensional equation
\cite{gao2016universal}.
Related spectral, matrix-based, and motif-based methods have also been developed to preserve selected structural features within a reduced description
\cite{laurence2019prx,thibeault2020threefold,burgio2021compphys}.
The usefulness of such reductions, however, depends on their accuracy, which is influenced by the interaction function, the network structure, and the dispersion of individual node states around the effective state
\cite{kundu2022accuracy}.
\vspace{0.3 cm}

\begin{minipage}{\linewidth}
    \centering
    \includegraphics[width=\linewidth]{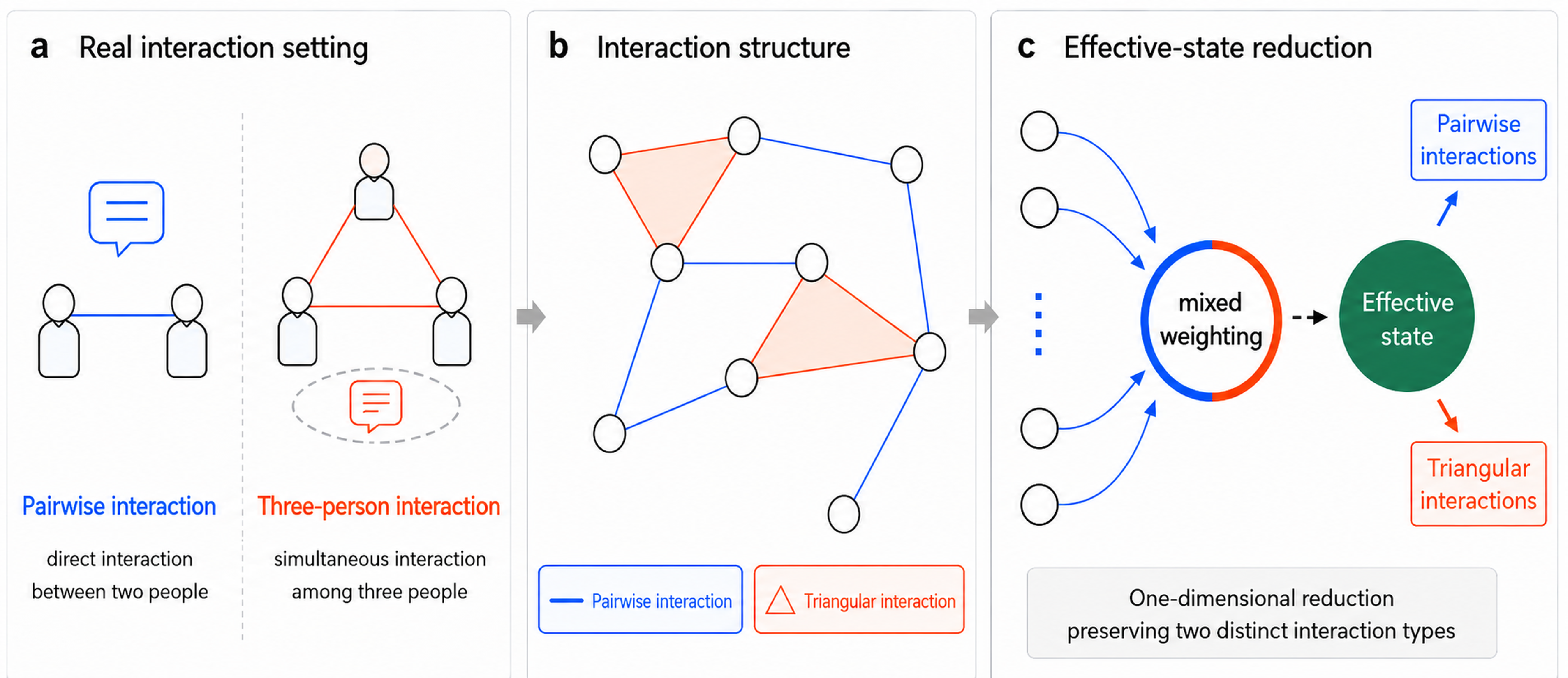}
    \captionof{figure}{{From mixed interactions to a one-dimensional effective-state description. (a) A connection between two individuals represents pairwise interaction, whereas a triangular interaction represents a simultaneous interaction among three individuals. (b) These two interaction types are mapped onto a simplicial network, with blue edges indicating pairwise connections and orange-filled triangles indicating three-node interactions. (c) The node states are combined through a mixed weighting procedure to define a single effective state. Although the system is reduced to one dimension, the pairwise and triangular contributions remain distinct in the reduced dynamics.}
    }
    \label{fig1}
\end{minipage}
\vspace{1em}

\par
A further limitation concerns the form of the interactions themselves. Standard network models represent interactions through edges, but pairwise coupling is not sufficient to describe many collective processes. The state of a unit may depend on the simultaneous states of several neighbouring units. Such higher-order interactions are commonly represented using hypergraphs, simplicial complexes, and related higher-order network structures
\cite{battiston2021natphys,boccaletti2023structure,bick2023higher,
Gracht2024,Majhi2022,Krishnagopal2023HigherOrderMultiplex}.
Their presence can change the qualitative nature of the dynamics. For example, spreading processes on simplicial complexes and hypergraphs may exhibit discontinuous transitions, bistability, and hysteresis
\cite{iacopini2019simplicial,arruda2020contagion,Fang2024SocialContagion}.
Higher-order interactions can also alter coexistence and stability in ecological systems
\cite{grilli2017higher},
and may produce epidemic thresholds that differ from those predicted by pairwise models
\cite{higham2021epidemics,sun2021hypergraph}.

These developments have motivated low-dimensional descriptions of systems with higher-order interactions. In particular, Ghosh \textit{et al.} developed a one-dimensional reduction for higher-order contagious processes and showed that an effective triangular interaction strength can capture epidemic thresholds, abrupt transitions, and bistable behaviour
\cite{ghosh2023chaos}.
Their result demonstrated that collective higher-order contagion can be represented through a compact scalar equation. More recently, Tiwari \textit{et al.} developed a dimension-reduction framework for dynamical networks with purely higher-order interactions
\cite{Tiwari2026RSPA}.
The problem considered in this paper extends this line of work in a different direction: pairwise and triangular interactions coexist and act simultaneously within the same node dynamics.
Existing effective-state reductions generally formulate the collective variable for one interaction order (pairwise or higher order) at a time. Extending such a description to a simplicial complex with coexisting pairwise and triangular interaction requires more than adding another coupling term to the scalar equation. The effective state must account for the fact that a node may be structurally important through its pairwise degree, its triangular degree, or both. Moreover, the pairwise and triangular contributions involve different neighbourhood averages and therefore introduce different closure errors. A useful reduction should preserve these distinctions while remaining low-dimensional.

\par
Here, we derive a one-dimensional effective-state reduction for nonlinear dynamics on simplicial complexes with coexisting pairwise and triangular interactions. Each node is assigned a mixed structural weight formed from its pairwise and triangular degrees. This construction defines a single effective state, while the reduced equation retains two distinct structural coefficients:  associated with both pairwise coupling and triangular coupling. A schematic representation of the mixed interaction structure is shown in Fig.~\ref{fig1}.
%The reduction is obtained through a fluctuation expansion about the effective state. This yields explicit closure conditions for the intrinsic dynamics, the pairwise interaction term, and the triangular interaction term. 
The reduction also identifies the terms neglected by the one-dimensional approximation and relates the resulting error to the dispersion of node states and their structural weighting. These diagnostics are important because close agreement between a reduced steady state and a particular network realisation does not, by itself, explain why the approximation works or under what conditions it may break down.

\par
We test the framework on three nonlinear models with distinct transition mechanisms: a gene-regulatory system
\cite{Meena2023stability,gao2016universal},
the double-well system
\cite{kundu2022rspa},
and the susceptible--infected--susceptible (SIS) model
\cite{ghosh2023chaos,Meena2023stability,pastor2015rmp}.
The gene-regulatory model combines degradation with Hill-type activation. The double-well model tests whether the reduction preserves two stable branches separated by an unstable branch, whereas the SIS model represents a threshold-driven spreading process with both pairwise and higher-order infection channels.

\par
For each model, we compare the full-network dynamics with the reduced equation using coupling-parameter sweeps, steady-state branch calculations, and progressive node-removal experiments on synthetic and real-world networks. We also evaluate the closure terms directly to relate the observed steady-state error to the approximations used in the derivation. The reduced model reproduces the principal branches and transition trends of all three systems. Agreement is strongest when the weighted node states remain concentrated around the effective state and becomes weaker when structural heterogeneity produces a broader state distribution. The resulting framework provides a compact description of resilience and critical transitions while retaining the separate roles of pairwise and higher-order interactions.

\section{Model}
\label{sec:model}

We consider a network of \(N\) nodes, where the state of node \(i\) is
denoted by \(x_i\). Each node evolves under intrinsic dynamics, pairwise
coupling, and triangular higher-order coupling:
\begin{equation}
\dot{x}_i
=
F(x_i)
+
D\sum_j A_{ij}G(x_i,x_j)
+
D_{\Delta}\sum_{j,l}A_{ijl}H(x_i,x_j,x_l),
\qquad
i=1,\ldots,N.
\label{eq:full_model}
\end{equation}
here, \(F\) denotes the intrinsic dynamics, \(G\) is the pairwise
interaction function, and \(H\) is the triangular higher-order interaction
function. The adjacency matrix \(A_{ij}\) represents pairwise interactions,
whereas \(A_{ijl}\) denotes the third-order adjacency tensor. The parameters
\(D\) and \(D_{\Delta}\) control the pairwise and higher-order coupling
strengths, respectively.
For each node, we define the pairwise degree and triangular degree as
\begin{equation}
k_i=\sum_j A_{ij},
\qquad
k_i^{\Delta}=\sum_{j,l}A_{ijl}.
\label{eq:degree_definition}
\end{equation}

To account for both structural contributions in the collective state, we
introduce the mixed structural weight
\begin{equation}
w_i=k_i+k_i^{\Delta}.
\label{eq:weight_definition}
\end{equation}
For any node-level quantity \(Y_i\), define the corresponding weighted
average by
\[
T_w(Y_i)
=
\frac{\sum_i w_iY_i}{\sum_i w_i}.
\]
The effective state is then defined as
\begin{equation}
x
=
x_{\mathrm{eff}}
=
T_w(x_i)
=
\frac{\sum_i w_ix_i}{\sum_i w_i}
=
\frac{\sum_i(k_i+k_i^{\Delta})x_i}
     {\sum_i(k_i+k_i^{\Delta})}.
\label{eq:effective_state}
\end{equation}
Thus, nodes with larger pairwise or triangular degree contribute more
strongly to the reduced state.
Differentiating Eq.~\eqref{eq:effective_state} and substituting the full
dynamics from Eq.~\eqref{eq:full_model}, we obtain
\begin{equation}
\dot{x}=T_F+T_G+T_H,
\label{eq:effective_dynamics}
\end{equation}
where
\begin{align*}
T_F
&=
\frac{\sum_i w_iF(x_i)}{\sum_iw_i},
\\
T_G
&=
\frac{
D\sum_iw_i\sum_jA_{ij}G(x_i,x_j)
}{
\sum_iw_i
},
\\
T_H
&=
\frac{
D_{\Delta}\sum_iw_i\sum_{j,l}A_{ijl}H(x_i,x_j,x_l)
}{
\sum_iw_i
}.
\end{align*}
The objective is to approximate these three weighted contributions using
only the effective state \(x\).

\subsection{Weighted Fluctuation Decomposition}

We express each node state as
\begin{equation}
x_i=x+\xi_i,
\qquad
x=x_{\mathrm{eff}},
\label{eq:fluctuation_decomposition}
\end{equation}
where \(\xi_i\) denotes the fluctuation of node \(i\) around the effective
state.

\begin{lemma}[Weighted fluctuation cancellation]
\label{lem:weighted_cancellation}
Let the mixed structural weight be \(w_i=k_i+k_i^{\Delta}\), and let the
effective state be defined by Eq.~\eqref{eq:effective_state}. Under the
decomposition in Eq.~\eqref{eq:fluctuation_decomposition}, the weighted
fluctuation satisfies $T_w(\xi_i)=0.$

\end{lemma}

\begin{proof}
Using \(x_i=x+\xi_i\), we have
\[
T_w(\xi_i)
=
T_w(x_i-x).
\]
Since \(x\) is independent of the node index \(i\),
\[
T_w(\xi_i)
=
T_w(x_i)-x.
\]
By the definition of the effective state,
\(T_w(x_i)=x\). Therefore,
\[
T_w(\xi_i)=x-x=0.
\]
\end{proof}
Equivalently, Lemma~\ref{lem:weighted_cancellation} gives the exact identity
\begin{equation}
\sum_i(k_i+k_i^{\Delta})\xi_i=0.
\label{eq:weighted_cancellation}
\end{equation}
This cancellation eliminates the first-order weighted fluctuation in the
intrinsic contribution and supports the leading-order closure developed
below.

\subsection{Leading-Order Closure}

We first consider the intrinsic contribution. Expanding \(F(x_i)\) around
the effective state \(x\), we obtain
\[
F(x_i)
=
F(x+\xi_i)
=
F(x)
+
F'(x)\xi_i
+
\frac{1}{2}F''(x)\xi_i^2
+
\mathcal{O}(\xi_i^3).
\]
Taking the weighted average gives
\[
T_F
=
F(x)
+
F'(x)T_w(\xi_i)
+
\frac{1}{2}F''(x)T_w(\xi_i^2)
+
\mathcal{O}\!\left(T_w(|\xi_i|^3)\right).
\]
Using Lemma~\ref{lem:weighted_cancellation}, the first-order term vanishes
exactly. Therefore, to leading order in the state fluctuations,
\[
T_F\approx F(x).
\]
For the pairwise interaction, we expand
\[
G(x_i,x_j)
=
G(x,x)
+
G_1(x,x)\xi_i
+
G_2(x,x)\xi_j
+
\mathcal{O}(\xi^2),
\]
where \(G_1\) and \(G_2\) denote the partial derivatives of \(G\) with
respect to its first and second arguments.
Under the leading-order pairwise closure
\[
G(x_i,x_j)\approx G(x,x),
\]
the pairwise contribution becomes
\[
T_G
\approx
D
\frac{
\sum_iw_i\sum_jA_{ij}
}{
\sum_iw_i
}
G(x,x).
\]
Using \(k_i=\sum_jA_{ij}\), this gives
\[
T_G
\approx
D
\frac{\sum_iw_ik_i}{\sum_iw_i}
G(x,x).
\]
Similarly, for the triangular interaction,
\[
H(x_i,x_j,x_l)
=
H(x,x,x)
+
H_1(x,x,x)\xi_i
+
H_2(x,x,x)\xi_j
+
H_3(x,x,x)\xi_l
+
\mathcal{O}(\xi^2),
\]
where \(H_1\), \(H_2\), and \(H_3\) denote the partial derivatives of \(H\)
with respect to its three arguments.
Under the leading-order higher-order closure
\[
H(x_i,x_j,x_l)\approx H(x,x,x),
\]
we obtain
\[
T_H
\approx
D_{\Delta}
\frac{
\sum_iw_i\sum_{j,l}A_{ijl}
}{
\sum_iw_i
}
H(x,x,x).
\]

Using \(k_i^{\Delta}=\sum_{j,l}A_{ijl}\), this becomes
\[
T_H
\approx
D_{\Delta}
\frac{\sum_iw_ik_i^{\Delta}}{\sum_iw_i}
H(x,x,x).
\]
Since \(w_i=k_i+k_i^{\Delta}\), the structural products satisfy
\[
w_ik_i
=
(k_i+k_i^{\Delta})k_i
=
k_i^2+k_ik_i^{\Delta},
\]
and
\[
w_ik_i^{\Delta}
=
(k_i+k_i^{\Delta})k_i^{\Delta}
=
k_ik_i^{\Delta}+(k_i^{\Delta})^2.
\]
The effective pairwise and higher-order structural coefficients are
therefore
\begin{equation}
\beta_G
=
\frac{
\sum_i k_i^2
+
\sum_i k_ik_i^{\Delta}
}{
\sum_i k_i
+
\sum_i k_i^{\Delta}
},
\qquad
\beta_H
=
\frac{
\sum_i k_ik_i^{\Delta}
+
\sum_i(k_i^{\Delta})^2
}{
\sum_i k_i
+
\sum_i k_i^{\Delta}
}.
\label{eq:effective_coefficients}
\end{equation}
Here, \(\beta_G\) represents the parameter dominated by pairwise structure, whereas
\(\beta_H\) represents the parameter dominated by triangular structure.
Substituting the leading-order approximations into
Eq.~\eqref{eq:effective_dynamics} gives the final reduced equation
\begin{equation}
\dot{x}
=
F(x)
+
D\beta_GG(x,x)
+
D_{\Delta}\beta_HH(x,x,x).
\label{eq:reduced_model}
\end{equation}

Thus, the original \(N\)-dimensional system is reduced to a single scalar
equation, while the pairwise and higher-order structural effects remain
separated through \(\beta_G\) and \(\beta_H\).

\subsection{Approximation Assumptions and Closure Conditions}

The reduction uses the following three leading-order closures:
\[
\mathrm{A}_1:\qquad T_F \approx F(x),
\]
\[
\mathrm{A}_2:\qquad G(x_i,x_j)\approx G(x,x),
\]
and
\[
\mathrm{A}_3:\qquad
H(x_i,x_j,x_l)\approx H(x,x,x).
\]
Condition \(\mathrm{A}_1\) approximates the weighted intrinsic contribution
by the intrinsic dynamics at the effective state. Conditions
\(\mathrm{A}_2\) and \(\mathrm{A}_3\) replace the pairwise and triangular
interactions by their values at the effective state.
The decomposition \(x_i=x+\xi_i\) and the exact identity
\[
\sum_i (k_i+k_i^{\Delta})\xi_i=0
\]
are used in the derivation but are not additional closure assumptions.
Therefore, the numerical analysis validates only
\(\mathrm{A}_1\)--\(\mathrm{A}_3\).
\section{Approximation Validation}

\subsection{Common Validation Notation}

For all three dynamical systems, we use the mixed structural weight
\[
w_i=k_i+k_i^{\Delta},
\]
and define the weighted average of a node-level quantity \(Y_i\) as
\[
T_w(Y_i)
=
\frac{\sum_i w_iY_i}{\sum_iw_i}.
\]

The effective state is
\[
x=x_{\mathrm{eff}}=T_w(x_i)
=
\frac{\sum_iw_ix_i}{\sum_iw_i},
\]
while the effective pairwise and higher-order structural coefficients are
\[
\beta_G=T_w(k_i)
=
\frac{\sum_iw_ik_i}{\sum_iw_i},
\qquad
\beta_H=T_w(k_i^{\Delta})
=
\frac{\sum_iw_ik_i^{\Delta}}{\sum_iw_i}.
\]

Writing
\[
x_i=x+\xi_i,
\]
the definition of \(x_{\mathrm{eff}}\) gives the exact identity
\[
T_w(\xi_i)=0.
\]

This identity is used in the derivation but is not treated as an
approximation. We validate only the three closures
\[
\mathrm{A}_1:\quad
T_F^{\mathrm{full}}\approx T_F^{\mathrm{red}},
\]
\[
\mathrm{A}_2:\quad
T_G^{\mathrm{full}}\approx T_G^{\mathrm{red}},
\]
and
\[
\mathrm{A}_3:\quad
T_H^{\mathrm{full}}\approx T_H^{\mathrm{red}}.
\]

In each diagnostic plot, the reduced term is placed on the \(x\)-axis and
the corresponding full-network weighted term on the \(y\)-axis. The
diagonal \(y=x\) represents exact agreement.

% ============================================================

\subsection{Approximation Validation for the Gene-Regulatory Model}

The full gene-regulatory dynamics are
\[
\dot{x}_i
=
-Bx_i^f
+
D\sum_jA_{ij}
\frac{x_j^h}{1+x_j^h}
+
\frac{D_{\Delta}}{2}
\sum_{j<l}A_{ijl}
\frac{(x_jx_l)^h}{1+(x_jx_l)^h}.
\]

Here,
\[
F(x_i)=-Bx_i^f,
\qquad
G(x_i,x_j)=\frac{x_j^h}{1+x_j^h},
\qquad
H(x_i,x_j,x_l)
=
\frac{(x_jx_l)^h}{1+(x_jx_l)^h}.
\]

The reduced equation is
\[
\dot{x}
=
-Bx^f
+
D\beta_G\frac{x^h}{1+x^h}
+
\frac{D_{\Delta}}{2}
\beta_H\frac{x^{2h}}{1+x^{2h}}.
\]

\subsubsection*{Intrinsic Closure}

The full and reduced intrinsic terms are
\[
T_F^{\mathrm{full}}
=
T_w(-Bx_i^f),
\qquad
T_F^{\mathrm{red}}
=
-Bx^f.
\]

For \(B\neq0\),
\[
R_{A_1}^{\mathrm{Gene}}
=
\frac{
T_w(x_i^f)
}{
x^f
}.
\]

The intrinsic closure is accurate when
\[
\left|
T_w(x_i^f)-x^f
\right|
\ll
|x^f|.
\]

For \(f=1\), it is exact because \(T_w(x_i)=x\).

\subsubsection*{Pairwise Closure}

The full and reduced pairwise terms are
\[
T_G^{\mathrm{full}}
=
D\,T_w\left(
\sum_jA_{ij}
\frac{x_j^h}{1+x_j^h}
\right),
\]
\[
T_G^{\mathrm{red}}
=
D\beta_G
\frac{x^h}{1+x^h}.
\]

For \(D\neq0\),
\[
R_{A_2}^{\mathrm{Gene}}
=
\frac{
T_w\left(
\sum_jA_{ij}
\frac{x_j^h}{1+x_j^h}
\right)
}{
\displaystyle
\beta_G\frac{x^h}{1+x^h}
}.
\]

The pairwise closure is accurate when
\[
\left|
T_w\left(
\sum_jA_{ij}
\frac{x_j^h}{1+x_j^h}
\right)
-
\beta_G\frac{x^h}{1+x^h}
\right|
\ll
\left|
\beta_G\frac{x^h}{1+x^h}
\right|.
\]

\subsubsection*{Higher-Order Closure}

The full and reduced higher-order terms are
\[
T_H^{\mathrm{full}}
=
\frac{D_{\Delta}}{2}
T_w\left(
\sum_{j<l}A_{ijl}
\frac{(x_jx_l)^h}
     {1+(x_jx_l)^h}
\right),
\]
\[
T_H^{\mathrm{red}}
=
\frac{D_{\Delta}}{2}
\beta_H
\frac{x^{2h}}{1+x^{2h}}.
\]

For \(D_{\Delta}\neq0\),
\[
R_{A_3}^{\mathrm{Gene}}
=
\frac{
T_w\left(
\sum_{j<l}A_{ijl}
\frac{(x_jx_l)^h}
     {1+(x_jx_l)^h}
\right)
}{
\displaystyle
\beta_H
\frac{x^{2h}}{1+x^{2h}}
}.
\]

The higher-order closure is accurate when
\[
\left|
T_w\left(
\sum_{j<l}A_{ijl}
\frac{(x_jx_l)^h}
     {1+(x_jx_l)^h}
\right)
-
\beta_H\frac{x^{2h}}{1+x^{2h}}
\right|
\ll
\left|
\beta_H\frac{x^{2h}}{1+x^{2h}}
\right|.
\]
% ============================================================
\subsection{Approximation Validation for the Double-Well Model}

The full double-well dynamics are
\[
\dot{x}_i
=
-(x_i-r_1)(x_i-r_2)(x_i-r_3)
+
D\sum_jA_{ij}x_j
+
\frac{D_{\Delta}}{2}
\sum_{j<l}A_{ijl}x_jx_l.
\]

Here,
\[
F(x_i)
=
-(x_i-r_1)(x_i-r_2)(x_i-r_3),
\qquad
G(x_i,x_j)=x_j,
\qquad
H(x_i,x_j,x_l)=x_jx_l.
\]

The reduced equation is
\[
\dot{x}
=
-(x-r_1)(x-r_2)(x-r_3)
+
D\beta_Gx
+
\frac{D_{\Delta}}{2}\beta_Hx^2.
\]

\subsubsection*{Intrinsic Closure}

Let
\[
a=r_1+r_2+r_3,
\qquad
b=r_1r_2+r_2r_3+r_3r_1,
\qquad
c=r_1r_2r_3.
\]

Then
\[
F(x)=-x^3+ax^2-bx+c.
\]

The full and reduced intrinsic terms are
\[
T_F^{\mathrm{full}}=T_w(F(x_i)),
\qquad
T_F^{\mathrm{red}}=F(x).
\]

Using \(x_i=x+\xi_i\) and \(T_w(\xi_i)=0\),
\[
T_F^{\mathrm{full}}
=
F(x)
+
(a-3x)T_w(\xi_i^2)
-
T_w(\xi_i^3).
\]

Therefore,
\[
R_{A_1}^{\mathrm{DW}}
=
1+
\frac{
(a-3x)T_w(\xi_i^2)-T_w(\xi_i^3)
}{
F(x)
}.
\]

The intrinsic closure is accurate when
\[
\left|
(a-3x)T_w(\xi_i^2)-T_w(\xi_i^3)
\right|
\ll
|F(x)|.
\]

\subsubsection*{Pairwise Closure}

The full and reduced pairwise terms are
\[
T_G^{\mathrm{full}}
=
D\,T_w\left(
\sum_jA_{ij}x_j
\right),
\qquad
T_G^{\mathrm{red}}
=
D\beta_Gx.
\]

For \(D\neq0\),
\[
R_{A_2}^{\mathrm{DW}}
=
1+
\frac{
T_w\left(
\sum_jA_{ij}\xi_j
\right)
}{
\beta_Gx
}.
\]

The pairwise closure is accurate when
\[
\left|
T_w\left(
\sum_jA_{ij}\xi_j
\right)
\right|
\ll
|\beta_Gx|.
\]

\subsubsection*{Higher-Order Closure}

The full and reduced higher-order terms are
\[
T_H^{\mathrm{full}}
=
\frac{D_{\Delta}}{2}
T_w\left(
\sum_{j<l}A_{ijl}x_jx_l
\right),
\qquad
T_H^{\mathrm{red}}
=
\frac{D_{\Delta}}{2}
\beta_Hx^2.
\]

For \(D_{\Delta}\neq0\),
\[
R_{A_3}^{\mathrm{DW}}
=
1+
\frac{
xT_w\left(
\sum_{j<l}A_{ijl}(\xi_j+\xi_l)
\right)
+
T_w\left(
\sum_{j<l}A_{ijl}\xi_j\xi_l
\right)
}{
\beta_Hx^2
}.
\]

The higher-order closure is accurate when
\[
\left|
xT_w\left(
\sum_{j<l}A_{ijl}(\xi_j+\xi_l)
\right)
+
T_w\left(
\sum_{j<l}A_{ijl}\xi_j\xi_l
\right)
\right|
\ll
|\beta_Hx^2|.
\]
\subsection{Approximation Validation for the SIS Model}

The full SIS dynamics are
\[
\dot{x}_i
=
-\mu x_i
+
\lambda(1-x_i)\sum_jA_{ij}x_j
+
\frac{\gamma}{2}(1-x_i)
\sum_{j<l}A_{ijl}x_jx_l.
\]

Here,
\[
F(x_i)=-\mu x_i,
\qquad
G(x_i,x_j)=(1-x_i)x_j,
\qquad
H(x_i,x_j,x_l)=(1-x_i)x_jx_l.
\]

The reduced equation is
\[
\dot{x}
=
-\mu x
+
\lambda\beta_G(1-x)x
+
\frac{\gamma}{2}\beta_H(1-x)x^2.
\]

\subsubsection*{Intrinsic Closure}

The full and reduced intrinsic terms are
\[
T_F^{\mathrm{full}}=T_w(-\mu x_i),
\qquad
T_F^{\mathrm{red}}=-\mu x.
\]

Since \(T_w(x_i)=x\),
\[
T_F^{\mathrm{full}}
=
-\mu T_w(x_i)
=
-\mu x
=
T_F^{\mathrm{red}}.
\]

Thus,
\[
R_{A_1}^{\mathrm{SIS}}
=
\frac{T_F^{\mathrm{full}}}
     {T_F^{\mathrm{red}}}
=
1.
\]

Therefore, the intrinsic closure is exact because the recovery term is
linear.

\subsubsection*{Pairwise Closure}

The full and reduced pairwise terms are
\[
T_G^{\mathrm{full}}
=
\lambda T_w\left(
(1-x_i)\sum_jA_{ij}x_j
\right),
\qquad
T_G^{\mathrm{red}}
=
\lambda\beta_G(1-x)x.
\]

Using \(x_i=x+\xi_i\) and \(x_j=x+\xi_j\),
\[
(1-x_i)\sum_jA_{ij}x_j
=
(1-x-\xi_i)
\left(
k_ix+\sum_jA_{ij}\xi_j
\right).
\]

Taking the weighted average gives
\begin{align*}
T_w\left(
(1-x_i)\sum_jA_{ij}x_j
\right)
={}&
\beta_G(1-x)x
+
(1-x)T_w\left(\sum_jA_{ij}\xi_j\right)
\\
&-
xT_w(k_i\xi_i)
-
T_w\left(
\xi_i\sum_jA_{ij}\xi_j
\right).
\end{align*}

Define
\begin{align*}
\mathcal{C}_{A_2}^{\mathrm{SIS}}
={}&
(1-x)T_w\left(\sum_jA_{ij}\xi_j\right)
-
xT_w(k_i\xi_i)
\\
&-
T_w\left(
\xi_i\sum_jA_{ij}\xi_j
\right).
\end{align*}

Then
\[
R_{A_2}^{\mathrm{SIS}}
=
1+
\frac{
\mathcal{C}_{A_2}^{\mathrm{SIS}}
}{
\beta_G(1-x)x
}.
\]

The pairwise closure is accurate when
\[
\left|
\mathcal{C}_{A_2}^{\mathrm{SIS}}
\right|
\ll
\left|
\beta_G(1-x)x
\right|.
\]

\subsubsection*{Higher-Order Closure}

The full and reduced higher-order terms are
\[
T_H^{\mathrm{full}}
=
\frac{\gamma}{2}
T_w\left(
(1-x_i)\sum_{j<l}A_{ijl}x_jx_l
\right),
\]
\[
T_H^{\mathrm{red}}
=
\frac{\gamma}{2}
\beta_H(1-x)x^2.
\]

Define
\[
U_i
=
\sum_{j<l}A_{ijl}(\xi_j+\xi_l),
\qquad
V_i
=
\sum_{j<l}A_{ijl}\xi_j\xi_l.
\]

Then
\[
\sum_{j<l}A_{ijl}x_jx_l
=
k_i^{\Delta}x^2+xU_i+V_i.
\]

Therefore,
\begin{align*}
T_w\left(
(1-x_i)\sum_{j<l}A_{ijl}x_jx_l
\right)
={}&
\beta_H(1-x)x^2
+
(1-x)xT_w(U_i)
\\
&+
(1-x)T_w(V_i)
-
x^2T_w(k_i^{\Delta}\xi_i)
\\
&-
xT_w(\xi_iU_i)
-
T_w(\xi_iV_i).
\end{align*}

Define
\begin{align*}
\mathcal{C}_{A_3}^{\mathrm{SIS}}
={}&
(1-x)xT_w(U_i)
+
(1-x)T_w(V_i)
\\
&-
x^2T_w(k_i^{\Delta}\xi_i)
-
xT_w(\xi_iU_i)
-
T_w(\xi_iV_i).
\end{align*}

Then
\[
R_{A_3}^{\mathrm{SIS}}
=
1+
\frac{
\mathcal{C}_{A_3}^{\mathrm{SIS}}
}{
\beta_H(1-x)x^2
}.
\]

The higher-order closure is accurate when
\[
\left|
\mathcal{C}_{A_3}^{\mathrm{SIS}}
\right|
\ll
\left|
\beta_H(1-x)x^2
\right|.
\]

% ============================================================
\section{Numerical Validation}
\label{sec:numerical_validation}

We validate the proposed one-dimensional reduction on three nonlinear
dynamical systems with coexisting pairwise and higher-order interactions:
a gene-regulatory system, a double-well system, and the
susceptible--infected--susceptible (SIS) model. These systems exhibit
distinct dynamical characteristics, including nonlinear activation and
degradation, bistability, and threshold-driven spreading. For each
system, the effective steady state obtained from the full
simplicial-complex dynamics is compared with the corresponding steady
state predicted by the reduced equation.
The numerical analysis is carried out in two stages. First, we study
the response of the full and reduced systems under variations of the
pairwise and higher-order interaction parameters. These parameter-sweep
experiments are performed separately on synthetic (Erd\H{o}s--R\'enyi (ER) and
Barab\'asi--Albert (BA) networks) and real world networks. The pairwise and higher-order
interaction parameters are varied individually, and two-parameter
heatmaps are additionally constructed by varying both parameters
simultaneously.
Second, we examine the system's resilience under progressive structural
perturbation. These experiments are performed on ER and BA networks
as well as on real-world networks. Nodes are progressively removed,
the remaining interaction structure is reconstructed after each
perturbation step, and the effective structural coefficients
$\beta_G$ and $\beta_H$ are recomputed. The effective steady states
of the full and reduced systems are then recorded throughout the
perturbation process.

The three dynamical models are considered in the following order:
the gene-regulatory system, the double-well system, and the SIS model.
For each system, we first present the parameter-sweep and heatmap
calculations, followed by the corresponding network-perturbation
analysis.

\subsection{Common Numerical Protocol}
\label{subsec:common_protocol}

The same numerical procedure is applied to all three dynamical systems.
Synthetic simplicial complexes are constructed from
Erd\H{o}s--R\'enyi (ER) and Barab\'asi--Albert (BA) networks. Unless
otherwise stated, the synthetic networks contain $N=500$ nodes. For
the ER networks, the connection probability is $p=0.1$, whereas for
the BA networks the growth parameter is $m=16$. Higher-order
interactions are defined on the triangular structures of the
corresponding networks.

For a given network and set of dynamical parameters, the full
node-level system is numerically integrated until a steady state is
reached. The resulting network state is then mapped to the effective
state using the mixed structural weighting introduced in the
reduction. The one-dimensional reduced equation is solved using the
same dynamical parameters together with the corresponding structural
coefficients $\beta_G$ and $\beta_H$. The full-network effective
steady state and the reduced steady state are then recorded for
comparison.

For the system-parameter analysis, the underlying network is kept
fixed while the interaction parameters (coupling strength) are varied.  First, the pairwise interaction
parameter is varied while the higher-order parameter is held fixed.
The higher-order interaction parameter is then varied while the
pairwise parameter is held fixed. These sweeps are carried out for
both ER and BA networks. For systems that admit multiple stable states,
the full and reduced models are initialized from both low and high
initial conditions so that the corresponding steady-state branches
can be followed.

The one-parameter sweeps are complemented by a two-parameter analysis
in which the pairwise and higher-order interaction parameters are
varied simultaneously. At each point in the parameter plane, the
one-dimensional reduced equation is solved from low and high initial
conditions, yielding the corresponding steady states
$x_{\mathrm{low}}$ and $x_{\mathrm{high}}$. The separation between
the two steady-state branches is quantified as
\begin{equation}
\Delta x
=
\left|
x_{\mathrm{high}}-x_{\mathrm{low}}
\right|.
\end{equation}
To facilitate comparison over the parameter plane, this quantity is
normalized according to
\begin{equation}
\Delta x_{\mathrm{norm}}
=
\frac{
\left|x_{\mathrm{high}}-x_{\mathrm{low}}\right|
}{
\displaystyle
\max_{D,D_{\Delta}}
\left|x_{\mathrm{high}}-x_{\mathrm{low}}\right|
}.
\end{equation}
The resulting normalized bistability strength is represented using a
heatmap, where values close to zero indicate convergence of both
initial conditions to the same steady state, while larger values
indicate a stronger separation between the two steady-state branches. The corresponding threshold contour is shown as a white dotted curve
in the two-parameter heatmaps and marks the numerical transition
between the bistable and monostable regimes.

Following the parameter-based calculations, node-removal experiments
are performed on ER, BA, and real-world networks. At each perturbation
step, $1\%$ of the nodes of the original network are removed, and the
largest connected component of the remaining network is retained.
The pairwise and triangular degrees are recomputed on the resulting
structure, from which the updated values of $\beta_G$ and $\beta_H$
are obtained. The full dynamical system is then evaluated on the
perturbed network, while the reduced equation is evaluated using the
corresponding updated structural coefficients.
This procedure generates a sequence of progressively perturbed
networks and allows the effective state to be tracked together with
the corresponding values of $\beta_G$ and $\beta_H$. The perturbation
process is continued until only $1\%$ of the original nodes remain or
the remaining network no longer contains a nontrivial connected
component. The full and reduced responses are recorded at each stage
of this sequence.

In addition to the dynamical response, we separately track the
evolution of the effective structural coefficients $\beta_G$ and
$\beta_H$ during node removal. This provides a direct description of
how progressive structural perturbation modifies the pairwise and
higher-order structural contributions entering the reduced equation.
The corresponding evolution of $\beta_G$ and $\beta_H$ for the ER,
BA, and real-world networks is presented in
Appendix~\ref{Appendix A}.

The approximations introduced in deriving the one-dimensional reduced
equation are also examined. In particular, the intrinsic, pairwise,
and higher-order terms computed from the full network are compared
with their corresponding reduced expressions. In addition, the
difference between the full-network effective steady state and the
reduced steady state is quantified over the considered parameter
range. These calculations provide a direct numerical diagnostic of
the individual approximation steps involved in the reduction. The
detailed diagnostic results for the gene-regulatory, double-well,
and SIS systems are presented in Appendix~\ref{Appendix B}.
\subsection{Gene-Regulatory System}
\label{subsec:gene_system}

We first consider the gene-regulatory model with coexisting pairwise
and higher-order interactions,
\begin{equation}
\dot{x}_i
=
-Bx_i^f
+
D\sum_j A_{ij}
\frac{x_j^h}{1+x_j^h}
+
\frac{D_{\Delta}}{2}
\sum_{j<l}A_{ijl}
\frac{(x_jx_l)^h}{1+(x_jx_l)^h},
\qquad i=1,\ldots,N.
\label{eq:gene_full}
\end{equation}
The corresponding one-dimensional reduced equation is
\begin{equation}
\dot{x}
=
-Bx^f
+
D\beta_G
\frac{x^h}{1+x^h}
+
\frac{D_{\Delta}}{2}\beta_H
\frac{x^{2h}}{1+x^{2h}}.
\label{eq:gene_reduced}
\end{equation}
Here, $x_i$ denotes the expression state of node $i$, while $x$
denotes the effective network state. The first term represents
degradation, with $B$ controlling the degradation strength and $f$
its nonlinearity. The second and third terms describe pairwise and
triangular activation, respectively, through Hill-type response
functions. The coefficients $\beta_G$ and $\beta_H$ encode the
effective pairwise and higher-order structures of the underlying
simplicial complex. Throughout the gene-regulatory analysis, we set
\[
B=1,
\qquad
f=1,
\qquad
h=2.
\]

The steady states of the reduced system satisfy
\[
-Bx^f
+
D\beta_G\frac{x^h}{1+x^h}
+
\frac{D_{\Delta}}{2}\beta_H
\frac{x^{2h}}{1+x^{2h}}
=0.
\]
Thus, $x=0$ is always a fixed point. For $x\neq 0$, using
$B=1$, $f=1$, and $h=2$, the nonzero steady states satisfy
\[
-1
+
D\beta_G\frac{x}{1+x^2}
+
\frac{D_{\Delta}}{2}\beta_H
\frac{x^3}{1+x^4}
=0.
\]
Equivalently,
\[
1
=
D\beta_G\frac{x}{1+x^2}
+
\frac{D_{\Delta}}{2}\beta_H
\frac{x^3}{1+x^4}.
\]
Multiplying by $(1+x^2)(1+x^4)$ gives
\[
(1+x^2)(1+x^4)
=
D\beta_Gx(1+x^4)
+
\frac{D_{\Delta}}{2}\beta_Hx^3(1+x^2).
\]
Therefore, the physically admissible nonzero fixed points are the
positive real roots of
\begin{equation}
x^6
-
\left(
D\beta_G+\frac{D_{\Delta}}{2}\beta_H
\right)x^5
+x^4
-\frac{D_{\Delta}}{2}\beta_Hx^3
+x^2
-D\beta_Gx
+1
=0.
\label{eq:gene_nonzero_fixed_points}
\end{equation}
Hence, the complete fixed-point structure consists of the zero fixed
point $x=0$ and the positive real roots of
Eq.~\eqref{eq:gene_nonzero_fixed_points}.

This system provides a nonlinear activation--degradation test of the
reduction. In contrast to an epidemic model, its steady-state
structure is determined by the competition between degradation and
interaction-induced activation. Depending on the interaction
parameters, the dynamics may support distinct steady states reached
from low and high initial conditions. We therefore integrate both the
full and reduced systems from the initial conditions $x_i(0)=0.01$
and $x_i(0)=10.0$ and compare their resulting steady-state branches.

% ============================================================
\subsubsection{System-Parameter Sweep}
\label{subsubsec:gene_parameter_sweep}
% ============================================================

We first examine the response of the gene-regulatory dynamics while
the underlying network is kept fixed. For both the ER and BA
simplicial complexes, the pairwise coupling strength $D$ is varied at
a fixed value of $D_{\Delta}$, after which $D_{\Delta}$ is varied at
a fixed value of $D$. The full-network effective steady state is
compared with the corresponding solution of
Eq.~\eqref{eq:gene_reduced} for both low and high initial conditions.

As shown in Fig.~\ref{fig:gene_system_parameter}(a) and
Fig.~\ref{fig:gene_system_parameter}(b), the reduced model reproduces
the principal steady-state branches obtained from the ER and BA
networks as the pairwise coupling strength is varied. In particular,
it follows both the low- and high-initial-condition solutions and
captures the parameter region over which the two branches remain
distinct. The corresponding $D_{\Delta}$ sweeps in
Fig.~\ref{fig:gene_system_parameter}(d) and
Fig.~\ref{fig:gene_system_parameter}(e) show the same overall
behavior under variations of the higher-order coupling strength.
Thus, the reduced equation preserves the main activation-induced
branch structure for both homogeneous and heterogeneous synthetic
networks.

To examine the joint influence of the two coupling mechanisms, the
reduced equation is also solved over the full
$(D,D_{\Delta})$ parameter plane. The normalized branch separation
$\Delta x_{\mathrm{norm}}$, defined in
Subsection~\ref{subsec:common_protocol}, is shown in
Fig.~\ref{fig:gene_system_parameter}(c) and
Fig.~\ref{fig:gene_system_parameter}(f). Regions with
$\Delta x_{\mathrm{norm}}$ close to zero correspond to parameter
values for which the low and high initial conditions converge to the
same steady state. Larger values identify regions in which the two
initial conditions converge to separated steady-state branches.

\vspace{1em}
\noindent
\begin{minipage}{\linewidth}
    \centering
    \includegraphics[width=0.90\linewidth]{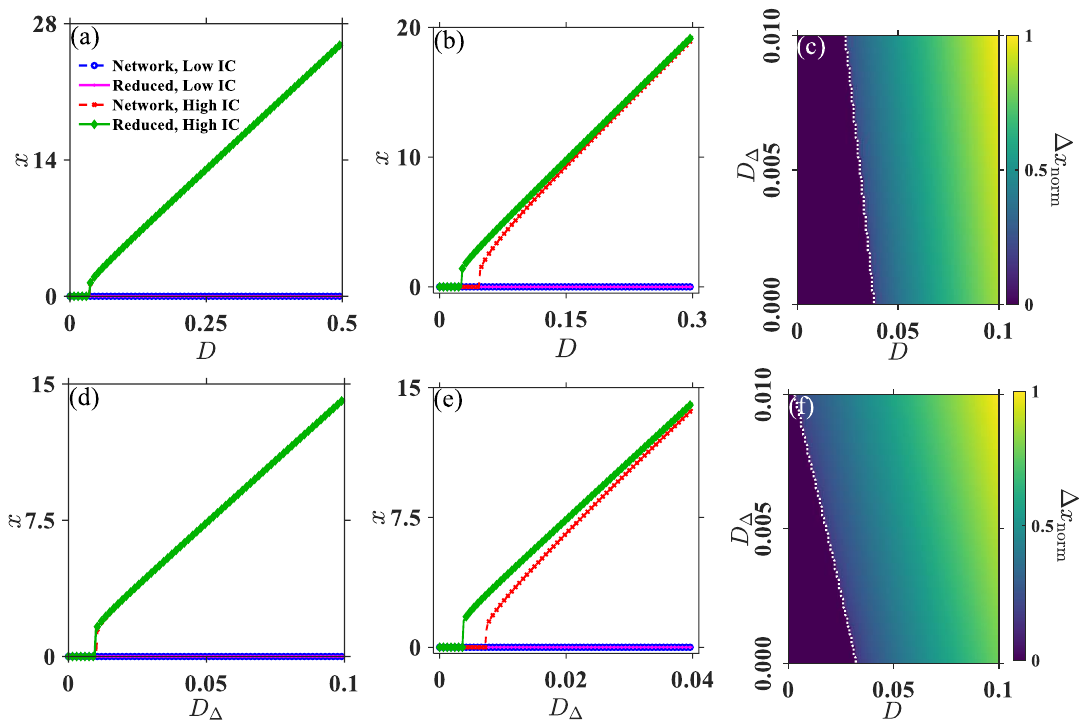}
    \captionof{figure}{
    System-parameter analysis of the gene-regulatory dynamics on ER
    and BA simplicial complexes. Panels (a) and (b) show the effective
    steady state $x$ as a function of the pairwise coupling strength
    $D$ for the ER and BA networks, respectively, with
    $D_{\Delta}$ held fixed. Panels (d) and (e) show the corresponding
    sweeps with respect to $D_{\Delta}$ at fixed $D$. Blue dashed
    curves with circles and magenta curves represent the full-network
    and reduced solutions obtained from low initial conditions,
    respectively. Red dashed curves and green curves with diamonds
    represent the corresponding high-initial-condition solutions.
    Panels (c) and (f) show the normalized bistability-strength maps
    obtained from the reduced model over the
    $(D,D_{\Delta})$ plane for the ER and BA networks, respectively.
    Values close to zero indicate convergence to the same steady
    state, whereas larger values indicate stronger separation between
    the low- and high-initial-condition branches. The white dotted
    contour marks the numerical transition between the bistable and
    monostable regimes.
    }
    \label{fig:gene_system_parameter}
\end{minipage}
\vspace{1em}

% ============================================================
\subsubsection{Network-Perturbation Test}
\label{subsubsec:gene_perturbation}
% ============================================================

We next examine whether the reduction remains accurate when the
underlying interaction structure is progressively degraded. The
node-removal protocol described in
Subsection~\ref{subsec:common_protocol} is applied to the ER and BA
networks and to a real biological network. At every perturbation step,
the largest connected component is retained, the pairwise and
triangular degrees are recalculated, and the corresponding effective
coefficients $\beta_G$ and $\beta_H$ are updated. The full dynamical
system is then integrated on the perturbed structure, while the
reduced equation is evaluated using the updated effective
coefficients.

For the real-network analysis, we use the
\href{https://networkrepository.com/bio-SC-LC.php}
{\texttt{bio-SC-LC} biological network}
from the Network Data Repository. The dataset represents a
gene-functional-association network and provides a heterogeneous
real-world structure against which the reduction can be tested beyond
the ER and BA ensembles. The same preprocessing and progressive
node-removal procedure are applied to this network as to the
synthetic networks.

Figure~\ref{fig:gene_perturbation}(a)--(c) shows the effective
steady states as functions of the evolving 
$\beta_G$ for the ER, BA, and \texttt{bio-SC-LC} networks,
respectively. Figure~\ref{fig:gene_perturbation}(d)--(f) shows the
corresponding states as functions of the 
$\beta_H$. In each case, the reduced solution follows the main
full-network response under progressive node removal and reproduces
the branch structure associated with the low and high initial
conditions. The agreement on the real network is particularly
important because its heterogeneous organization differs from the
controlled structures of the synthetic ensembles.

The results also show how the two effective structural coefficients
change along the perturbation sequence. Their evolution reflects the
simultaneous loss of pairwise connections and triangular interaction
structures as nodes are removed. Consequently, the perturbation test
evaluates not only the dynamical accuracy of the reduced equation but
also whether the updated coefficients $\beta_G$ and $\beta_H$
adequately summarize the changing interaction structure.

\vspace{1em}
\noindent
\begin{minipage}{\linewidth}
    \centering
    \includegraphics[width=0.90\linewidth]{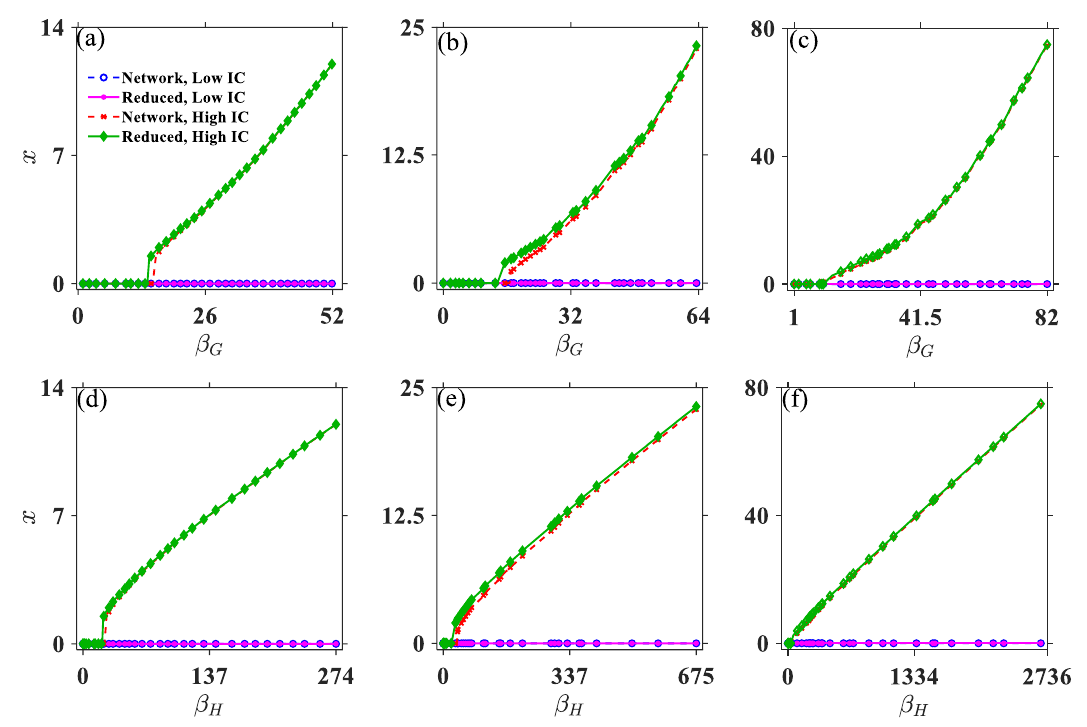}
    \captionof{figure}{
    Gene-regulatory dynamics under progressive network perturbation.
    Each point represents one stage of the node-removal process after
    retaining the largest connected component and recomputing the
    effective structural coefficients. Panels (a)--(c) show the
    full-network effective steady state and reduced steady state as
    functions of the effective pairwise coefficient $\beta_G$,
    whereas panels (d)--(f) show the corresponding states as
    functions of the effective higher-order coefficient $\beta_H$.
    The three columns correspond, from left to right, to the ER
    network, the BA network, and the
    \texttt{bio-SC-LC} real biological network. Blue dashed curves
    with circles and magenta curves denote the full-network and
    reduced solutions obtained from low initial conditions,
    respectively. Red dashed curves and green curves with diamonds
    denote the corresponding high-initial-condition solutions. The
    figure therefore compares the full and reduced dynamical
    responses while simultaneously tracking the structural changes
    produced by progressive node removal.
    }
    \label{fig:gene_perturbation}
\end{minipage}
\vspace{1em}
\subsection{Double-Well System}
\label{subsec:double_well_system}

We next consider the double-well system with coexisting pairwise and
higher-order interactions,
\begin{equation}
\dot{x}_i
=
-(x_i-r_1)(x_i-r_2)(x_i-r_3)
+
D\sum_j A_{ij}x_j
+
\frac{D_{\Delta}}{2}
\sum_{j<l}A_{ijl}x_jx_l,
\qquad i=1,\ldots,N.
\label{eq:double_well_full}
\end{equation}
The corresponding one-dimensional reduced equation is
\begin{equation}
\dot{x}
=
-(x-r_1)(x-r_2)(x-r_3)
+
D\beta_G x
+
\frac{D_{\Delta}}{2}\beta_H x^2.
\label{eq:double_well_reduced}
\end{equation}
Here, the cubic intrinsic term generates the local double-well
dynamics, while the second and third terms represent the pairwise and
triangular interaction contributions, respectively. The coefficients
$\beta_G$ and $\beta_H$ encode the effective pairwise and
higher-order structures of the underlying simplicial complex. We use
\[
r_1=1,
\qquad
r_2=2,
\qquad
r_3=5.
\]
In the absence of coupling, $r_1$ and $r_3$ correspond to the two
stable local equilibria, whereas $r_2$ is the unstable equilibrium
separating their basins of attraction.

The steady states of the reduced system satisfy
\[
-(x-r_1)(x-r_2)(x-r_3)
+
D\beta_Gx
+
\frac{D_{\Delta}}{2}\beta_Hx^2
=0.
\]
Using $r_1=1$, $r_2=2$, and $r_3=5$, this becomes
\[
-(x-1)(x-2)(x-5)
+
D\beta_Gx
+
\frac{D_{\Delta}}{2}\beta_Hx^2
=0.
\]
Since
\[
(x-1)(x-2)(x-5)
=
x^3-8x^2+17x-10,
\]
the steady-state equation can be written as
\[
-x^3
+
\left(
8+\frac{D_{\Delta}}{2}\beta_H
\right)x^2
+
\left(D\beta_G-17\right)x
+10
=0.
\]
Multiplying by $-1$, the fixed points are the real roots of
\begin{equation}
x^3
-
\left(
8+\frac{D_{\Delta}}{2}\beta_H
\right)x^2
+
\left(17-D\beta_G\right)x
-10
=0.
\label{eq:double_well_fixed_points}
\end{equation}

The double-well system provides a direct test of whether the reduction
can preserve bistability and the associated dependence on initial
conditions. The full-network and reduced systems are therefore
integrated from the low and high initial conditions $x_i(0)=0.01$ and
$x_i(0)=5.5$, respectively. These two initial conditions allow the
lower and upper steady-state branches to be followed separately as
the pairwise and higher-order coupling strengths are varied.
% ============================================================
\subsubsection{System-Parameter Sweep}
\label{subsubsec:double_well_parameter_sweep}
% ============================================================

We first examine the double-well dynamics while the underlying
network structure is kept fixed. For both the ER and BA simplicial
complexes, the pairwise coupling strength $D$ is varied while
$D_{\Delta}$ is held fixed. The higher-order coupling strength
$D_{\Delta}$ is subsequently varied while $D$ is held fixed. At every
parameter value, the effective steady state obtained from the full
network is compared with the corresponding prediction of
Eq.~\eqref{eq:double_well_reduced} for both low and high initial
conditions.

Figures~\ref{fig:double_well_system_parameter}(a) and
\ref{fig:double_well_system_parameter}(b) show the steady-state
responses of the ER and BA networks as the pairwise coupling strength
is varied. The reduced equation follows the principal full-network
branches and reproduces the separation between the states reached
from the low and high initial conditions. The corresponding
higher-order coupling sweeps in
Figs.~\ref{fig:double_well_system_parameter}(d) and
\ref{fig:double_well_system_parameter}(e) show that the reduced model
also captures the change in the branch structure generated by
triangular interactions. The comparison therefore tests both the
location of the steady-state branches and their persistence under
changes in the two interaction mechanisms.

To characterize the joint effect of $D$ and $D_{\Delta}$, the reduced
equation is additionally solved over the full
$(D,D_{\Delta})$ parameter plane. The normalized branch separation
$\Delta x_{\mathrm{norm}}$, defined in
Subsection~\ref{subsec:common_protocol}, is shown in
Figs.~\ref{fig:double_well_system_parameter}(c) and
\ref{fig:double_well_system_parameter}(f). Values close to zero
indicate that the low and high initial conditions converge to the
same steady state, whereas larger values indicate a stronger
separation between the two steady-state branches. 

\vspace{1em}
\noindent
\begin{minipage}{\linewidth}
    \centering
    \includegraphics[width=0.88\linewidth]{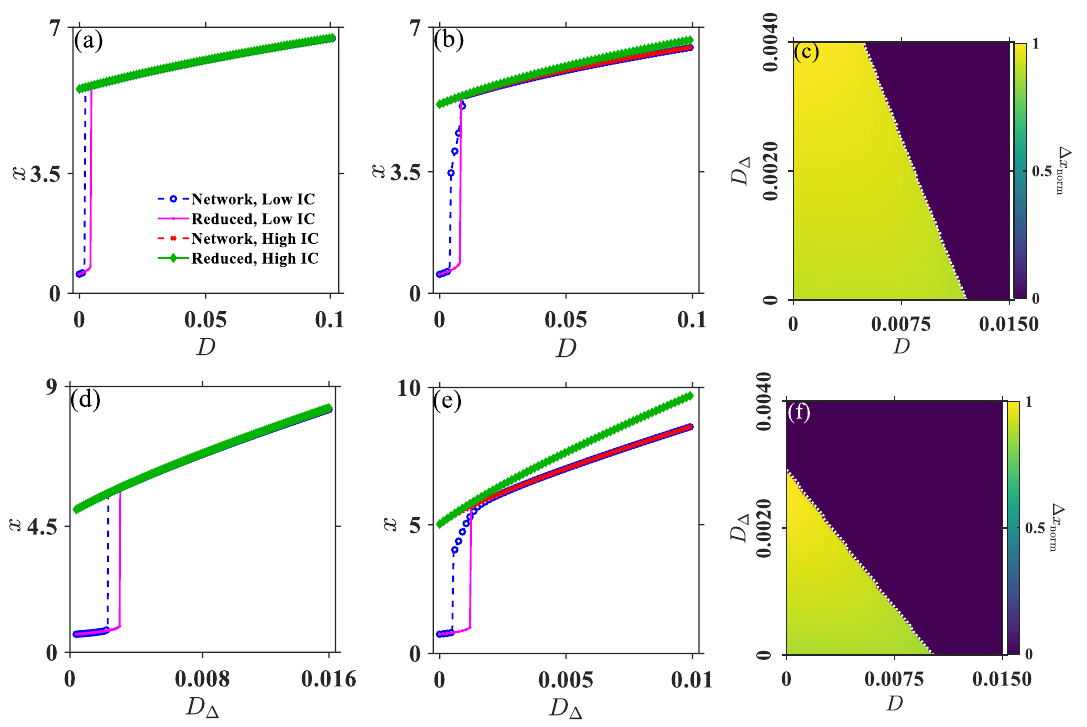}
    \captionof{figure}{
    System-parameter analysis of the double-well dynamics on ER and
    BA simplicial complexes. Panels (a) and (b) show the effective
    steady state $x$ as a function of the pairwise coupling strength
    $D$ for the ER and BA networks, respectively, with
    $D_{\Delta}$ held fixed. Panels (d) and (e) show the corresponding
    sweeps with respect to $D_{\Delta}$ at fixed $D$. Blue dashed
    curves with circles and magenta curves represent the full-network
    and reduced solutions obtained from low initial conditions,
    respectively. Red dashed curves and green curves with diamonds
    represent the corresponding high-initial-condition solutions.
    Panels (c) and (f) show the normalized bistability-strength maps
    obtained from the reduced model over the $(D,D_{\Delta})$
    parameter plane for the ER and BA networks, respectively. Values
    close to zero indicate convergence of the low and high initial
    conditions to the same steady state, whereas larger values
    indicate stronger separation between the two steady-state
    branches. The white dotted contour marks the numerical transition
    between the bistable and monostable regimes.
    }
    \label{fig:double_well_system_parameter}
\end{minipage}
\vspace{1em}

% ============================================================
\subsubsection{Network-Perturbation Test}
\label{subsubsec:double_well_perturbation}
% ============================================================

We next examine the performance of the double-well reduction under
progressive structural perturbation. The node-removal protocol
described in Subsection~\ref{subsec:common_protocol} is applied to
the ER and BA networks and to a real social network. At every
perturbation step, the largest connected component is retained, the
pairwise and triangular degrees are recomputed, and the updated
effective coefficients $\beta_G$ and $\beta_H$ are obtained. The
full dynamical system is then integrated on the perturbed network,
while the reduced equation is evaluated using the corresponding
updated structural coefficients.

For the real-network analysis, we use the
\href{https://networkrepository.com/socfb-Caltech36.php}
{\texttt{socfb-Caltech36}} network from the Facebook100 collection.
This dataset is an undirected and unweighted university friendship
network in which nodes represent individuals and edges represent
Facebook friendship relations. As in the synthetic-network experiments, the largest
connected component is retained during the perturbation process.
The Caltech36 network provides a structurally heterogeneous
real-world test that differs from both the homogeneous ER ensemble
and the preferentially grown BA ensemble.

Figures~\ref{fig:double_well_perturbation}(a)--(c) show the effective
steady states as functions of  
$\beta_G$ for the ER, BA, and Caltech36 networks, respectively.
Figures~\ref{fig:double_well_perturbation}(d)--(f) show the
corresponding responses as functions of $\beta_H$. In each case, the low- and
high-initial-condition branches are tracked throughout the
node-removal sequence. The reduced equation reproduces the principal
full-network response and follows the structural evolution of the two
branches as pairwise links and triangular interaction structures are
progressively removed.

The dependence on $\beta_G$ and $\beta_H$ also clarifies how the
effective interaction structure changes during perturbation. In
particular, the later stages of node removal can produce rapid changes
in the higher-order coefficient because triangular structures are
lost as the network becomes increasingly sparse. The insets in
Figs.~\ref{fig:double_well_perturbation}(d)--(f) enlarge the
small-$\beta_H$ region and make the corresponding late-stage branch
transition more clearly visible.

\vspace{1em}
\noindent
\begin{minipage}{\linewidth}
    \centering
    \includegraphics[width=0.90\linewidth]{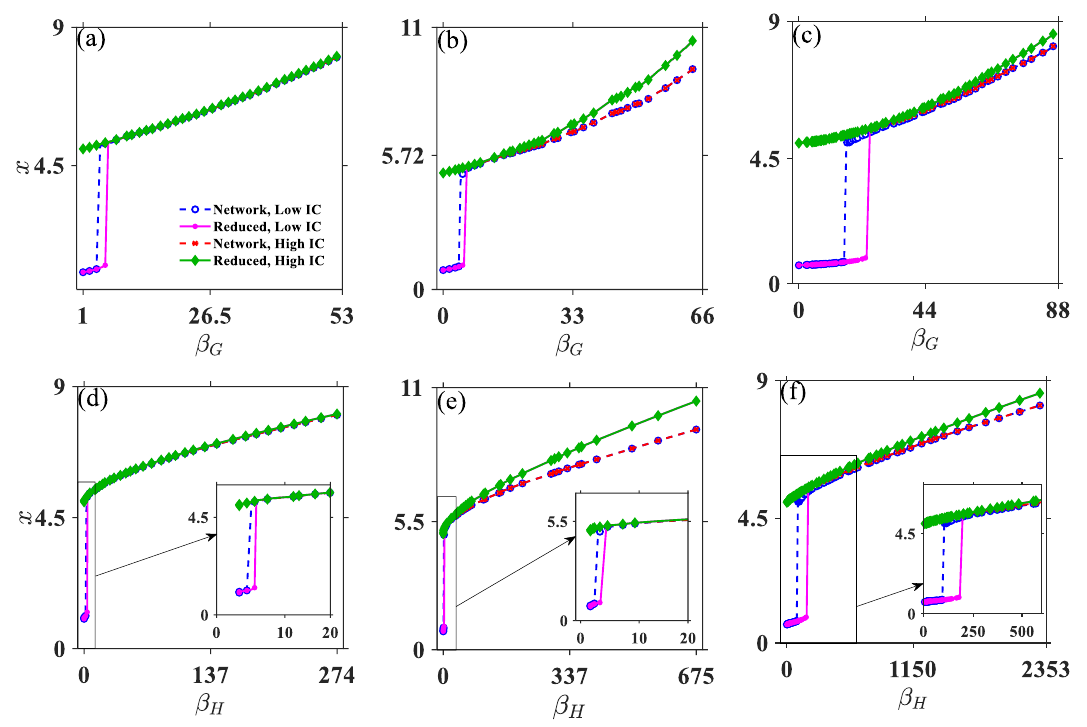}
    \captionof{figure}{
    Double-well dynamics under progressive network perturbation.
    Each point represents one stage of the node-removal process after
    retaining the largest connected component and recomputing the
    effective structural coefficients. Panels (a)--(c) show the
    full-network effective steady state and reduced steady state as
    functions of the effective pairwise coefficient $\beta_G$,
    whereas panels (d)--(f) show the corresponding states as
    functions of the effective higher-order coefficient $\beta_H$.
    The three columns correspond, from left to right, to the ER
    network, the BA network, and the
    \texttt{socfb-Caltech36} real social network. Blue dashed curves
    with circles and magenta curves denote the full-network and
    reduced solutions obtained from low initial conditions,
    respectively. Red dashed curves and green curves with diamonds
    denote the corresponding high-initial-condition solutions. The
    insets in panels (d)--(f) enlarge the small-$\beta_H$ region so
    that the branch transition occurring during the later stages of
    structural perturbation can be resolved more clearly.
    }
    \label{fig:double_well_perturbation}
\end{minipage}
\vspace{1em}
\subsection{SIS Model}
\label{subsec:sis_system}

Finally, we consider the susceptible--infected--susceptible (SIS)
model with coexisting pairwise and higher-order infection channels,
\begin{equation}
\dot{x}_i
=
-\mu x_i
+
\lambda(1-x_i)
\sum_j A_{ij}x_j
+
\frac{\gamma}{2}(1-x_i)
\sum_{j<l}A_{ijl}x_jx_l,
\qquad i=1,\ldots,N.
\label{eq:sis_full}
\end{equation}
The corresponding one-dimensional reduced equation is
\begin{equation}
\dot{x}
=
-\mu x
+
\lambda\beta_G(1-x)x
+
\frac{\gamma}{2}\beta_H(1-x)x^2.
\label{eq:sis_reduced}
\end{equation}
Here, $x_i$ denotes the infected state of node $i$, and $x$ denotes
the corresponding effective network state. The parameter $\mu$ is
the recovery rate, $\lambda$ controls transmission through pairwise
interactions, and $\gamma$ controls transmission through triangular
higher-order interactions. The factor $(1-x_i)$ accounts for the
susceptible fraction available for infection.

The steady states of the reduced system satisfy
\[
-\mu x
+
\lambda\beta_G(1-x)x
+
\frac{\gamma}{2}\beta_H(1-x)x^2
=0.
\]
Factoring out $x$ gives
\[
x\left[
-\mu
+
\lambda\beta_G(1-x)
+
\frac{\gamma}{2}\beta_Hx(1-x)
\right]
=0.
\]
Thus, $x=0$ is always a fixed point. For $x\neq 0$, the nonzero
fixed points satisfy
\[
-\mu
+
\lambda\beta_G(1-x)
+
\frac{\gamma}{2}\beta_Hx(1-x)
=0.
\]
Expanding and rearranging gives
\[
\frac{\gamma}{2}\beta_Hx^2
+
\left(
\lambda\beta_G-\frac{\gamma}{2}\beta_H
\right)x
+
\mu-\lambda\beta_G
=0.
\]
Therefore, for $\gamma\beta_H>0$, the two nonzero fixed-point
branches are
\begin{equation}
x_{\pm}
=
\frac{
\frac{\gamma}{2}\beta_H-\lambda\beta_G
\pm
\sqrt{
\left(
\lambda\beta_G+\frac{\gamma}{2}\beta_H
\right)^2
-
2\gamma\beta_H\mu
}
}{
\gamma\beta_H
}.
\label{eq:sis_nonzero_roots}
\end{equation}
These nonzero fixed points are real when
\[
\left(
\lambda\beta_G+\frac{\gamma}{2}\beta_H
\right)^2
\geq
2\gamma\beta_H\mu.
\]
Hence, the reduced SIS system always possesses the disease-free fixed
point $x=0$. Real nonzero roots can occur when the above discriminant
condition is satisfied, while physically admissible endemic states
must additionally satisfy $0<x\leq 1$.

The SIS model provides a threshold-driven benchmark for the proposed
reduction. The disease-free state competes with endemic states
generated by pairwise and higher-order transmission. In particular,
the nonlinear higher-order infection term may produce parameter
regions in which the long-term state depends on the initial infected
fraction. We therefore integrate both the full-network system and the
reduced equation from low and high initial conditions and compare the
resulting effective steady-state branches. Since the recovery term is
linear, discrepancies between the full and reduced dynamics arise
primarily from the approximations applied to the pairwise and
higher-order infection terms.

% ============================================================
\subsubsection{System-Parameter Sweep}
\label{subsubsec:sis_parameter_sweep}
% ============================================================

We first examine the SIS dynamics while keeping the underlying
simplicial complex fixed. For both the ER and BA networks, the
pairwise infection strength $\lambda$ is varied while the
higher-order infection strength $\gamma$ is held fixed. A
complementary sweep is then performed by varying $\gamma$ at a fixed
value of $\lambda$. At every parameter value, the effective
steady state obtained from the full-network dynamics is compared with
the corresponding prediction of Eq.~\eqref{eq:sis_reduced}. Both low
and high initial conditions are considered so that any dependence on
the initial infected fraction can be identified.

Figures~\ref{fig:sis_system_parameter}(a) and
\ref{fig:sis_system_parameter}(b) show the response to variations in
the pairwise infection strength for the ER and BA networks,
respectively. The reduced equation follows the main full-network
response and reproduces the transition from low infection levels to
the endemic branch. The corresponding higher-order infection sweeps
in Figs.~\ref{fig:sis_system_parameter}(d) and
\ref{fig:sis_system_parameter}(e) demonstrate how triangular
transmission modifies the effective endemic state. The comparison of
the low- and high-initial-condition solutions additionally reveals
the parameter ranges over which the asymptotic state depends on the
initial infection level.

To examine the combined effects of the two transmission channels, the
reduced equation is solved over the full $(\lambda,\gamma)$ parameter
plane. Figures~\ref{fig:sis_system_parameter}(c) and
\ref{fig:sis_system_parameter}(f) show the normalized separation
$\Delta x_{\mathrm{norm}}$ between the steady states reached from the
high and low initial conditions for the ER and BA networks,
respectively. The definition and normalization of this quantity are
given in Subsection~\ref{subsec:common_protocol}. Values close to zero
indicate that the two initial conditions converge to the same
steady state, whereas larger values identify parameter regions with
a stronger separation between the two asymptotic branches.

\vspace{1em}
\noindent
\begin{minipage}{\linewidth}
    \centering
    \includegraphics[width=0.90\linewidth]{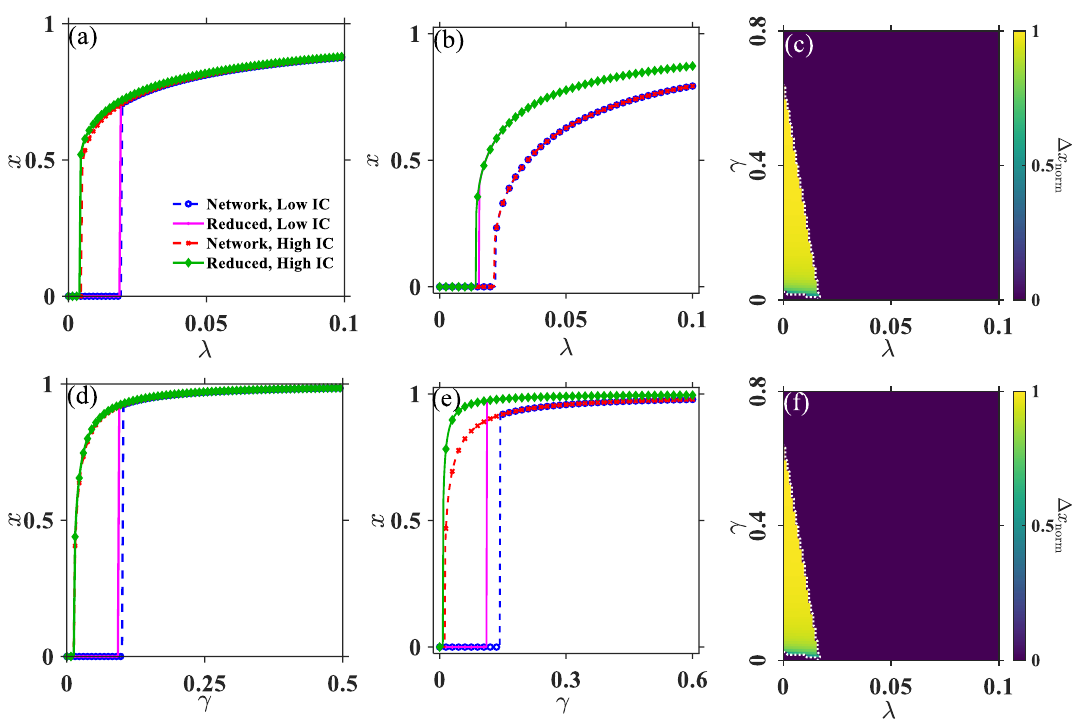}
    \captionof{figure}{
    System-parameter analysis of the SIS dynamics on ER and BA
    simplicial complexes. Panels (a) and (b) show the effective
    infected state $x$ as a function of the pairwise infection
    strength $\lambda$ for the ER and BA networks, respectively, with
    the higher-order infection strength $\gamma$ held fixed. Panels
    (d) and (e) show the corresponding sweeps with respect to
    $\gamma$ at fixed $\lambda$. Blue dashed curves with circles and
    magenta curves represent the full-network and reduced solutions
    obtained from low initial conditions, respectively. Red dashed
    curves and green curves with diamonds represent the corresponding
    solutions obtained from high initial conditions. Panels (c) and
    (f) show the normalized steady-state branch separation over the
    $(\lambda,\gamma)$ parameter plane for the ER and BA networks,
    respectively. Values close to zero indicate convergence of the
    low and high initial conditions to the same steady state, whereas
    larger values indicate stronger separation between the resulting
    branches. The white dotted contour marks the numerical transition
    between the bistable and monostable regimes.
    }
    \label{fig:sis_system_parameter}
\end{minipage}
\vspace{1em}

% ============================================================
\subsubsection{Network-Perturbation Test}
\label{subsubsec:sis_perturbation}
% ============================================================

We next test the SIS reduction under progressive structural
perturbation. The node-removal procedure described in
Subsection~\ref{subsec:common_protocol} is applied to the ER and BA
networks and to a real human-contact network. At each perturbation
step, the largest connected component is retained, the pairwise and
triangular degrees are recomputed, and the updated effective
coefficients $\beta_G$ and $\beta_H$ are obtained. The full SIS
dynamics is evaluated on the perturbed network, while the reduced
equation is solved using the corresponding updated structural
coefficients.

For the real-network analysis, we use the
\href{https://networkrepository.com/ia-infect-dublin.php}
{\texttt{ia-infect-dublin} human-contact network}
available through the Network Data Repository. In this network, nodes
represent individuals and edges represent physical-proximity contacts
recorded during the Infectious SocioPatterns event at the Science
Gallery in Dublin. The contact-based nature of the dataset makes it a
suitable real-world structure for examining SIS-type spreading beyond
the synthetic ER and BA ensembles. The same preprocessing,
largest-connected-component selection, and progressive node-removal
protocol are applied to this network.

Figures~\ref{fig:sis_perturbation}(a)--(c) show the full-network
effective infected state and the reduced prediction as functions of
 $\beta_G$ for the ER, BA, and
\texttt{ia-infect-dublin} networks, respectively.
Figures~\ref{fig:sis_perturbation}(d)--(f) show the corresponding
dependence on  $\beta_H$. These plots
therefore describe both the dynamical response to structural
degradation and the accompanying changes in the effective pairwise
and triangular interaction strengths.

Across the perturbation sequence, comparison of the full and reduced
curves determines whether the two effective structural coefficients
continue to provide an adequate low-dimensional description as nodes
and interaction structures are progressively removed. The low- and
high-initial-condition branches also indicate whether the reduction
preserves any dependence of the endemic state on the initial infected
fraction during network degradation.

\vspace{1em}
\noindent
\begin{minipage}{\linewidth}
    \centering
    \includegraphics[width=0.90\linewidth]{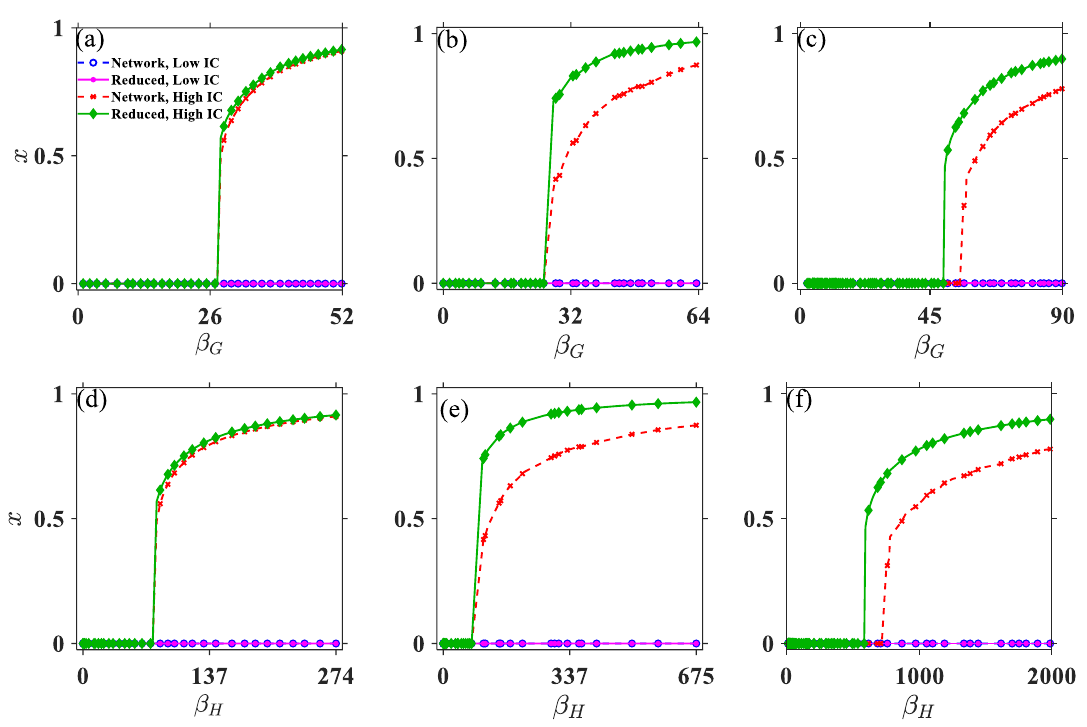}
    \captionof{figure}{
    SIS dynamics under progressive network perturbation. Each point
    represents one stage of the node-removal process after retaining
    the largest connected component and recomputing the effective
    structural coefficients. Panels (a)--(c) show the full-network
    effective infected state and reduced steady state as functions of
    the effective pairwise coefficient $\beta_G$, whereas panels
    (d)--(f) show the corresponding states as functions of the
    effective higher-order coefficient $\beta_H$. The three columns
    correspond, from left to right, to the ER network, the BA network,
    and the \texttt{ia-infect-dublin} human-contact network. Blue
    dashed curves with circles and magenta curves denote the
    full-network and reduced solutions obtained from low initial
    conditions, respectively. Red dashed curves and green curves with
    diamonds denote the corresponding high-initial-condition
    solutions. The figure thus compares the full and reduced SIS
    responses while tracking the structural changes generated by
    progressive node removal.
    }
    \label{fig:sis_perturbation}
\end{minipage}
\vspace{1em}
\section{Discussion}
\label{sec:discussion}

We developed a one-dimensional effective-state reduction for nonlinear
dynamics on simplicial complexes with simultaneous pairwise and
triangular interactions. Unlike reductions in which a higher-order
term is appended to a pairwise formulation, the collective variable
is defined through the mixed structural weight
$w_i=k_i+k_i^{\Delta}$. The reduced one dimensional system consist of two structure dependent parameter $\beta_G$ (dominated by the pairwise degree) and $\beta_H$ (dominated by the higher-order degree).
An SSM-based method has recently been proposed to construct a nonlinear invariant manifold from dominant spectral modes and reconstruct node-level dynamics through a lifting map \cite{Bhaskaran2026SSM}. 
The pairwise coefficient depends on the
pairwise-degree second moment and the cross-moment between pairwise
and triangular degrees, whereas the higher-order coefficient depends
on the triangular-degree second moment and the same cross-moment.
Thus, the reduction retains information about how the two interaction
orders are distributed over the same set of nodes.

The fluctuation analysis clarifies the approximation mechanism.
Because the effective state is defined using the mixed weight, the
first-order weighted fluctuation vanishes exactly, eliminating the
linear correction from the weighted intrinsic term. The remaining
corrections involve correlations between state deviations, local
neighbourhoods, and structural degrees. Structural heterogeneity
therefore becomes dynamically important when it generates broad or
structure-dependent node-state distributions, rather than merely
because the degree distribution is wide. Simplicial complexes with
similar $\beta_G$ and $\beta_H$ may consequently display different
reduction errors if their state fluctuations are organised
differently.

The three dynamical systems highlight different aspects of this
mechanism. In the SIS model, the linear recovery term is reduced
exactly, so the approximation error originates from the pairwise and
higher-order infection closures. These closures remain highly
accurate for the ER network, while the BA network exhibits larger
systematic deviations, particularly in the triangular contribution.
Nevertheless, the reduced equation preserves the principal epidemic
branch structure and provides a close approximation to the effective
steady state.

For the gene-regulatory system, the intrinsic closure is also exact
for $f=1$, and the main approximations arise from the nonlinear Hill
activation terms. The ER results show close agreement except near
states close to zero, where a small absolute discrepancy can produce
a comparatively large relative error. In the BA network, the
pairwise activation closure deviates more visibly than the triangular
closure, while the overall reduction becomes more accurate when the
higher-order contribution dominates the steady-state balance. This
shows that reduction accuracy depends not only on total interaction
strength but also on the relative dynamical importance of the two
interaction channels.

The double-well system provides a stronger test because its intrinsic
dynamics are cubic. The corresponding closure contains corrections
associated with the weighted variance and third central moment of the
node states. These corrections remain small for the ER network and
become more visible for the structurally heterogeneous BA network,
particularly on the high-state branch. The pairwise and higher-order
closures may also develop biases in opposite directions, producing
partial error compensation. Therefore, the term-wise diagnostics
provide useful information beyond the final effective-state error and
help identify how different approximation contributions combine.

The parameter sweeps and two-parameter maps show that the reduced
equation reproduces the principal steady-state branches and the
transition regions associated with activation, epidemic thresholds,
and bistability. The largest discrepancies occur near branch
transitions, where small changes in coupling strengths or structural
coefficients can shift the system between different stable states.
Error peaks in these regions therefore reflect both closure accuracy
and the inherent sensitivity of the local bifurcation structure.

The node-removal experiments provide a complementary structural test.
At each step, the largest connected component is retained and both
$\beta_G$ and $\beta_H$ are recalculated. Because edges and triangles
may be destroyed at different rates, retaining the two coefficients
separately allows the reduced equation to follow the changing balance
between pairwise and higher-order interactions. The agreement
observed over much of the perturbation sequence demonstrates that the
effective-state description remains informative across a broad range
of progressively changing network structures.

Within the present study, resilience is interpreted through the
persistence, displacement, and loss of effective steady-state
branches under dynamical and structural perturbations. The reduction
therefore provides a compact dynamical representation for tracking
branch survival, transition locations, and the changing roles of
pairwise and triangular interactions. This interpretation complements,
rather than replaces, alternative scalar measures of network
resilience.

The comparison between ER and BA networks further identifies the
conditions under which the reduction is most effective. ER networks
have relatively narrow structural distributions and more homogeneous
node states, leading to strong agreement at both the closure and
steady-state levels. BA networks offer a more demanding test because
of their broader pairwise and triangular participation patterns.
Although systematic closure deviations become more visible, the
scalar model continues to reproduce the principal branch structure.
This indicates that the framework remains useful even beyond the
most homogeneous structural regime.

The present formulation also provides several natural opportunities
for refinement. The additive weight $w_i=k_i+k_i^{\Delta}$ is
transparent, analytically convenient, and directly incorporates both
pairwise and triangular participation into the collective coordinate.
Future work may investigate data-driven or system-specific procedures
for selecting effective coordinates that further improve closure
accuracy in strongly heterogeneous networks. Similarly, systems
containing clearly separated communities, degree classes, or
localised modes may benefit from a two- or multi-dimensional
effective-state description.

The framework can also be extended beyond triangular interactions,
static structures, and steady states. Simplicial complexes containing
multiple interaction orders could be represented through one
effective coefficient for each relevant simplex order. Transient
responses, temporal simplicial complexes, and adaptive interactions
would further motivate time-dependent structural coefficients and
closure relations. These extensions provide a direct route for
applying the present construction to a broader class of higher-order
dynamical systems.

Finally, the neglected fluctuation terms suggest a systematic route
toward analytical error estimates based on weighted variances,
covariances, and neighbourhood correlations. Future perturbation
studies may also compare random node removal with targeted removal
based on pairwise degree, triangular degree, or mixed centrality.
Overall, the results demonstrate that mixed pairwise and higher-order
dynamics can be represented by a scalar equation without combining
both interaction orders into a single structural quantity. The
reduction is especially accurate when the structurally weighted node
states remain concentrated around the effective state, while the
separate closure diagnostics provide a practical way to identify and
address the sources of discrepancy when greater heterogeneity is
present.\\
\noindent\textbf{Declaration of AI use.} 
During the preparation of this work, the authors used language refinement tools to refine the English style of the presentation, alongside grammatical corrections. After using this tool/service, the authors reviewed and edited the content as needed and take full responsibility for the content of the publication.

\noindent\textbf{Acknowledgement:} PK acknowledge BSES Rajdhani Power Limited and BSES Yamuna Power Limited for the CSR grants to carry out the work at the Smart Energy Learning Centre (SELC), Dhirubhai Ambani University (DAU), Gandhinagar, Gujarat, India (Grant Number: CSR-25/BSES/A7-PRK/SELC). CH acknowledges support from ARNF India (Grant Number ANRF/ECRG/2024/000207/PMS).
\section{Appendix}
\subsection*{Appendix A. Evolution of Effective Structural Coefficients under Node Removal}
\label{Appendix A}
To further characterize the structural changes induced by node removal, we track $\beta_G$ and $\beta_H$, throughout the perturbation process. Both coefficients generally decrease as nodes are progressively removed, although their rates of decay vary across networks. In particular, $\beta_H$ often shows a stronger relative decline, indicating a stronger disruption of higher-order interaction structure under progressive network degradation.
\vspace{1em} % Adds a little space above
\noindent
\newpage
\begin{minipage}{\linewidth}
    \centering

    \includegraphics[width=.90\linewidth]{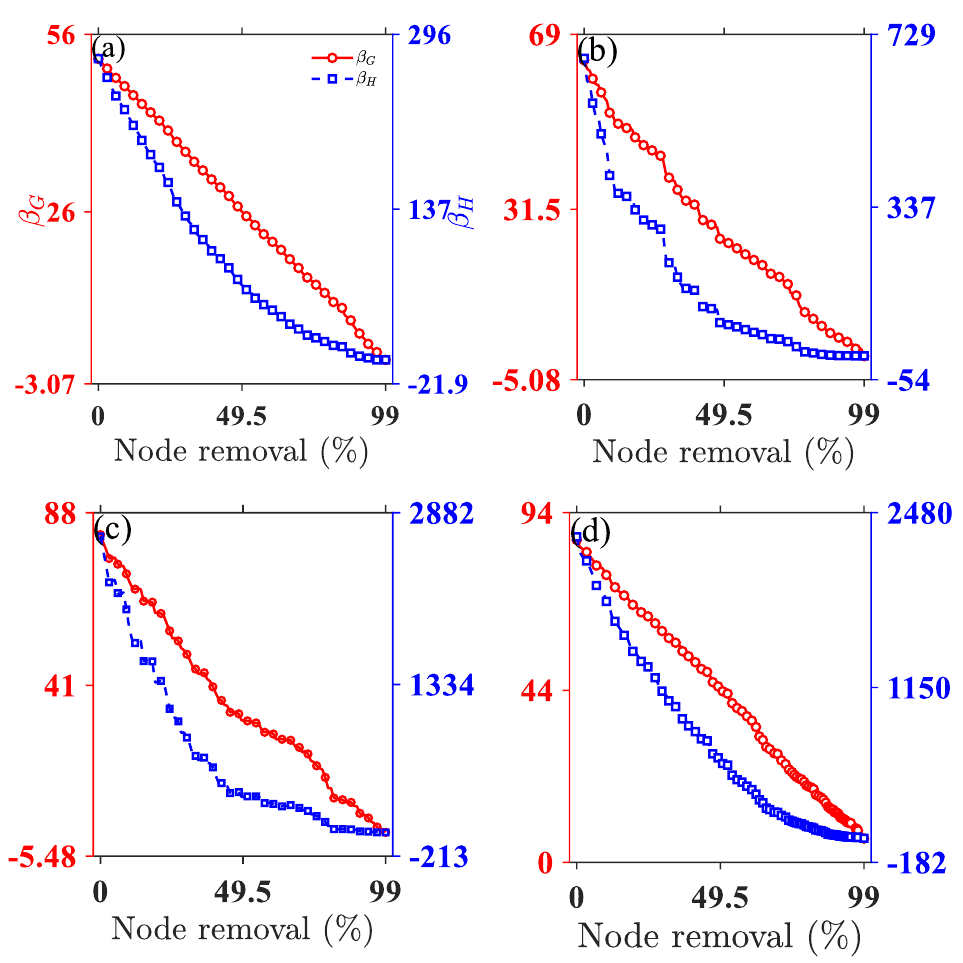}
   \captionof{figure}{Evolution of the effective pairwise coefficient $\beta_G$ (red circles, left axis) and the effective higher-order coefficient $\beta_H$ (blue squares, right axis) under progressive node removal. Panels (a) and (b) correspond to the ER and BA synthetic networks, respectively, while panels (c) and (d) correspond to the real networks bio-SC-LC and ia-infect-dublin, respectively.}
    \label{fig:placeholder}
\end{minipage}

\subsection*{Appendix B. Approximation-Diagnostic Validation}
\label{Appendix B}

This appendix examines the individual approximations entering the
one-dimensional reduction. For each dynamical system, the full-network
quantities associated with the intrinsic, pairwise, and higher-order
terms are compared with their corresponding reduced expressions.
Agreement with the diagonal indicates that the relevant approximation
is accurate. The final panel in each figure reports the relative error
between the full-network effective steady state and the reduced steady
state as the effective higher-order structural coefficient $\beta_H$
changes.

% ============================================================
\subsubsection*{B.1. Gene-Regulatory System}
% ============================================================

\vspace{1em}
\noindent
\begin{minipage}{\linewidth}
    \centering
    \includegraphics[width=0.80\linewidth]
    {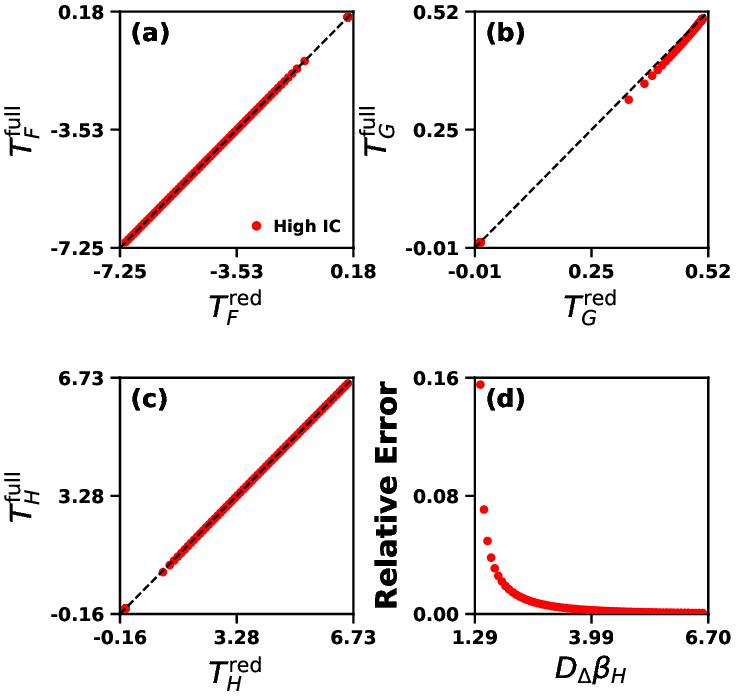}

    \captionof{figure}{
    Approximation validation for the ER gene-regulatory network on the
    high-initial-condition steady-state branch. Panels (a)--(c)
    compare the full-network quantities entering the intrinsic,
    pairwise, and higher-order terms with their corresponding reduced
    approximations. The intrinsic term remains close to the diagonal,
    indicating that the degradation contribution is accurately
    represented. The pairwise closure also follows the diagonal,
    although a small systematic deviation appears at larger values.
    The higher-order closure remains close to the diagonal throughout
    most of the tested range. Panel (d) shows the relative error
    between the full-network effective steady state and the reduced
    steady state as a function of the effective higher-order
    structural coefficient $D_\Delta\beta_H$. The error is small over most of
    the range, while the isolated large value occurs when the
    full-network effective steady state is close to zero, making the
    denominator in the relative-error measure very small.
    }
    \label{fig:gene_er_approximation}
\end{minipage}

\newpage

\vspace{1em}
\noindent
\begin{minipage}{\linewidth}
    \centering
    \includegraphics[width=0.80\linewidth]
    {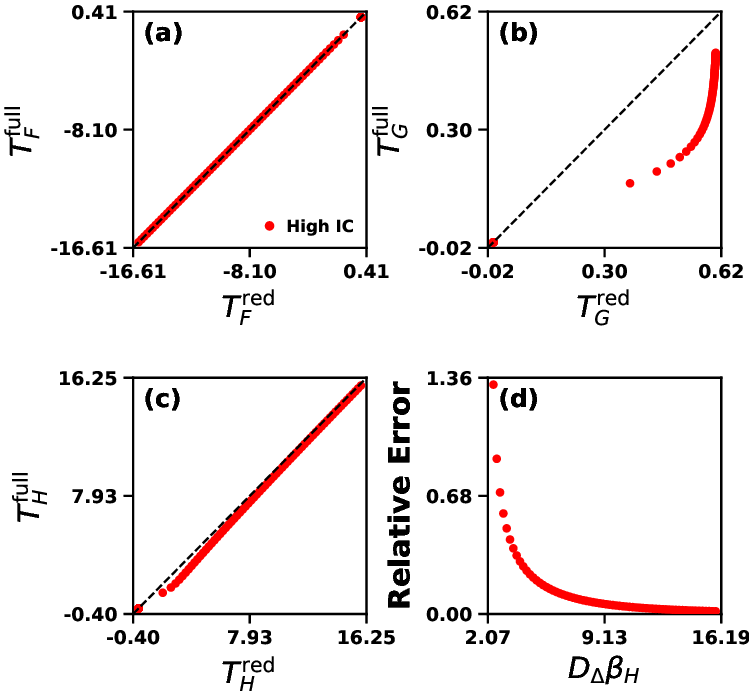}

    \captionof{figure}{
    Approximation validation for the BA gene-regulatory network with
    $m=16$ on the high-initial-condition steady-state branch. Panels
    (a)--(c) compare the full-network quantities entering the
    intrinsic, pairwise, and higher-order terms with their
    corresponding reduced approximations. The intrinsic approximation
    remains close to the diagonal because the degradation term is
    linear for the present choice $f=1$. The pairwise closure shows a
    systematic deviation below the diagonal, indicating that the
    reduced pairwise Hill-activation term is larger than the
    corresponding full-network contribution over part of the tested
    range. In contrast, the higher-order closure remains close to the
    diagonal for most parameter values. Panel (d) shows the relative
    error between the full-network effective steady state and the
    reduced steady state as a function of $D_\Delta\beta_H$. The error is
    larger in the weak higher-order regime and decreases as the
    higher-order structural contribution becomes stronger.
    }
    \label{fig:gene_ba_approximation}
\end{minipage}

% ============================================================
\subsubsection*{B.2. Double-Well System}
% ============================================================

\vspace{1em}
\noindent
\begin{minipage}{\linewidth}
    \centering
    \includegraphics[width=0.80\linewidth]
    {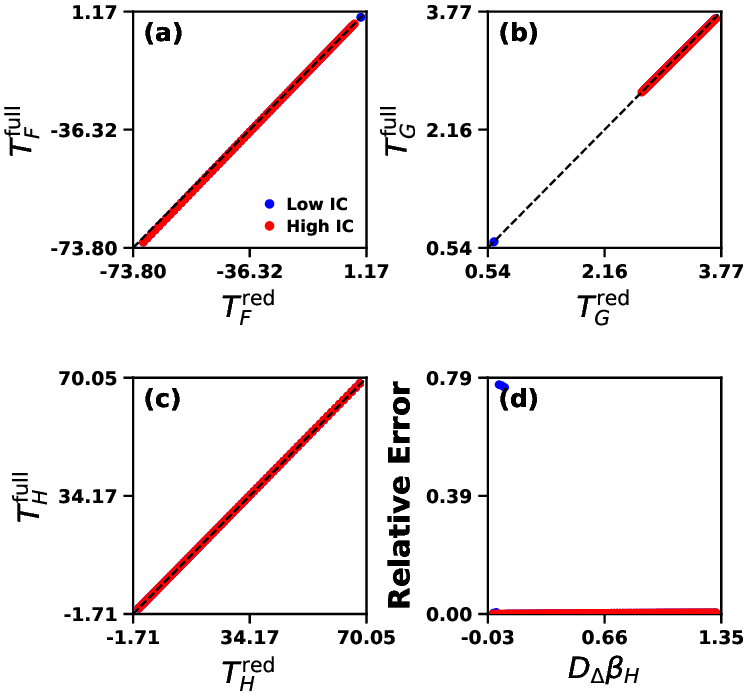}

    \captionof{figure}{
    Approximation validation for the ER double-well network under low
    and high initial conditions. Blue markers correspond to the
    low-initial-condition branch with $x_i(0)=0.01$, while red markers
    correspond to the high-initial-condition branch with
    $x_i(0)=5.5$. Panels (a)--(c) compare the full-network quantities
    entering the intrinsic, pairwise, and higher-order terms with
    their corresponding reduced approximations. The points remain
    close to the diagonal in all three panels, indicating that the
    intrinsic, pairwise, and higher-order contributions are accurately
    represented by the reduced description. The agreement is
    particularly strong for the higher-order closure in panel (c).
    Panel (d) shows the relative error between the full-network
    effective steady state and the reduced steady state as a function
    of the effective higher-order structural coefficient $D_\Delta\beta_H$.
    The relative error remains small over the tested range for both
    initial-condition branches, confirming the accuracy of the
    reduction for the ER double-well network.
    }
    \label{fig:dw_er_approximation}
\end{minipage}

\vspace{1em}

\noindent
\begin{minipage}{\linewidth}
    \centering
    \includegraphics[width=0.80\linewidth]
    {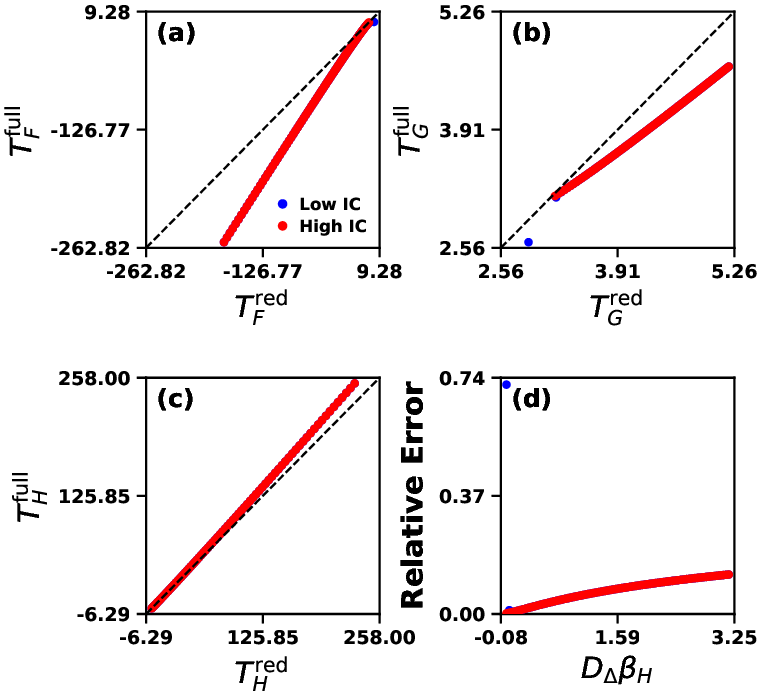}

    \captionof{figure}{
    Approximation validation for the BA double-well network with
    $m=16$ under low and high initial conditions. Blue markers
    correspond to the low-initial-condition branch with
    $x_i(0)=0.01$, while red markers correspond to the
    high-initial-condition branch with $x_i(0)=5.5$. Panels (a)--(c)
    compare the full-network quantities entering the intrinsic,
    pairwise, and higher-order terms with their corresponding reduced
    approximations. Relative to the ER case, stronger deviations from
    the diagonal are visible because of the greater structural
    heterogeneity of the BA network. The intrinsic approximation
    shows a noticeable deviation on the high-initial-condition branch,
    reflecting the nonlinear double-well term under heterogeneous node
    states. The pairwise closure lies mainly below the diagonal,
    whereas the higher-order closure lies above the diagonal at larger
    values. These trends indicate that the reduced pairwise term tends
    to exceed the corresponding full-network contribution, while the
    reduced higher-order term tends to underestimate it in part of the
    tested range. Panel (d) shows the relative error between the
    full-network effective steady state and the reduced steady state
    as a function of $D_\Delta\beta_H$. The error is larger than in the ER
    case, particularly on portions of the high-initial-condition
    branch, reflecting the stronger effect of structural and state
    heterogeneity in the BA network.
    }
    \label{fig:dw_ba_approximation}
\end{minipage}

% ============================================================
\subsubsection*{B.3. SIS Model}
% ============================================================

\vspace{1em}
\noindent
\begin{minipage}{\linewidth}
    \centering
    \includegraphics[width=0.80\linewidth]
    {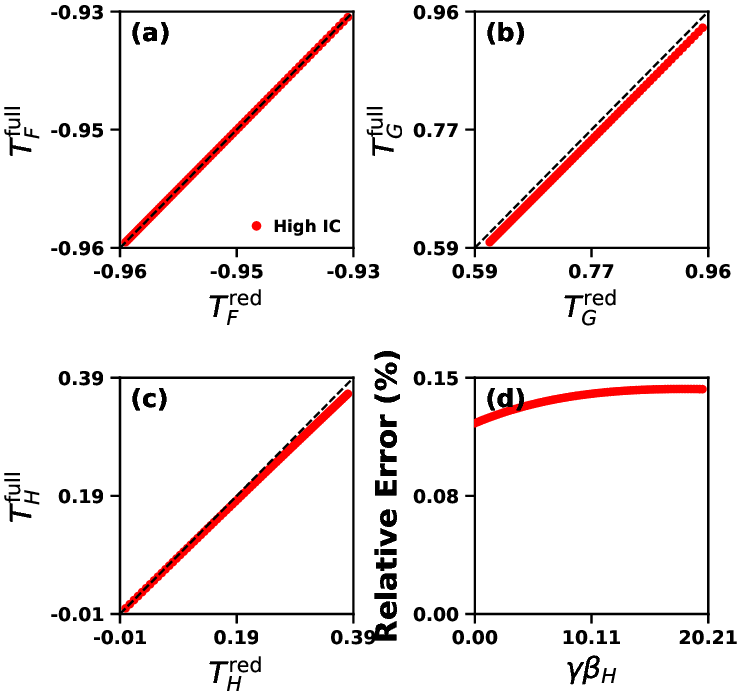}

    \captionof{figure}{
    Approximation validation for the ER SIS network on the
    high-initial-condition steady-state branch. Panels (a)--(c)
    compare the full-network quantities entering the intrinsic,
    pairwise, and higher-order terms with their corresponding reduced
    approximations. The intrinsic approximation is nearly exact,
    consistent with the linear recovery term. The pairwise and
    higher-order closures also remain close to the diagonal, with only
    small systematic deviations over the tested range. Panel (d)
    shows the relative error between the full-network effective steady
    state and the reduced steady state as a function of the effective
    higher-order structural coefficient  $\gamma\beta_H$. The small relative
    error confirms that the reduced description accurately reproduces
    the effective SIS steady state for the ER network.
    }
    \label{fig:sis_er_approximation}
\end{minipage}

\newpage

\vspace{1em}
\noindent
\begin{minipage}{\linewidth}
    \centering
    \includegraphics[width=0.75\linewidth]
    {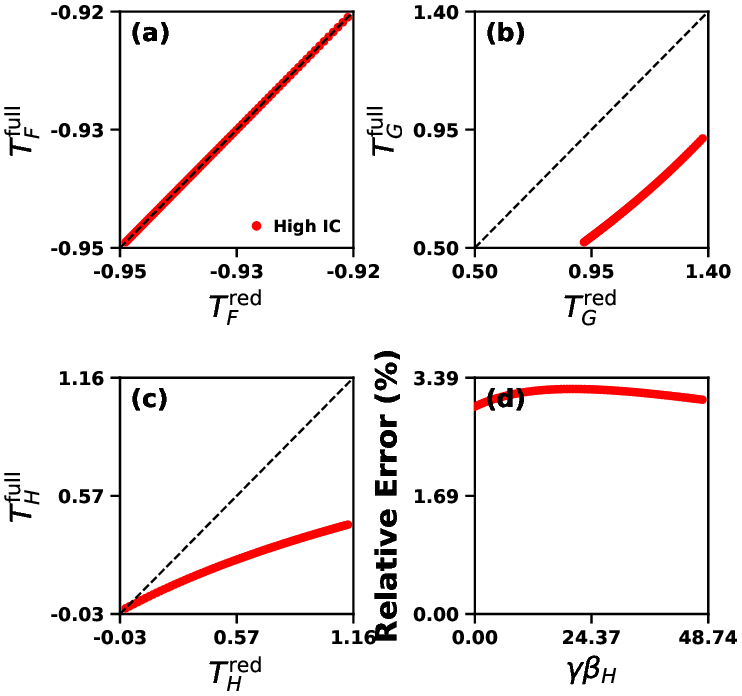}

    \captionof{figure}{
    Approximation validation for the BA SIS network with $m=16$ on
    the high-initial-condition steady-state branch. Panels (a)--(c)
    compare the full-network quantities entering the intrinsic,
    pairwise, and higher-order terms with their corresponding reduced
    approximations. The intrinsic approximation remains nearly exact
    because the SIS recovery term is linear. The pairwise and
    higher-order closures show systematic deviations below the
    diagonal, indicating that the corresponding reduced interaction
    terms exceed the full-network contributions over part of the
    tested range. The deviation is more pronounced for the
    higher-order closure, consistent with the stronger degree and
    triangular-structure heterogeneity of the BA network. Panel (d)
    shows the relative error between the full-network effective steady
    state and the reduced steady state as a function of $\gamma\beta_H$.
    Although the error is larger than for the ER network, the reduced
    model continues to reproduce the principal effective steady-state
    behavior of the BA SIS system.
    }
    \label{fig:sis_ba_approximation}
\end{minipage}

\bibliographystyle{unsrtnat}
\bibliography{references}
\end{document}